\documentclass[aps,prl,reprint,nofootinbib,twocolumn,superscriptaddress,10pt]{revtex4-1}
\pdfoutput=1
\usepackage[ascii]{inputenc}
\usepackage[bookmarksnumbered,hypertexnames=false]{hyperref}
\usepackage{bbm,microtype,mathrsfs,amsmath,amssymb,amsfonts,color,amsthm,graphicx,Quantum,bm,enumitem}
\usepackage[capitalize]{cleveref}
\usepackage[T1]{fontenc}
\usepackage[USenglish]{babel}
\usepackage[dvipsnames]{xcolor}
\usepackage{tikz,tikz-cd}
\usepackage{times}
\usetikzlibrary{through,calc,positioning,patterns}
\newtheorem{thm}{Theorem}\crefname{thm}{Theorem}{Theorems}
\newtheorem{lem}[thm]{Lemma}\crefname{lem}{Lemma}{Lemmas}
\crefname{prop}{Proposition}{Propositions}
\newtheorem{corr}[thm]{Corollary}\crefname{corr}{Corollary}{Corollaries}
\newcommand{\CH}{\mathcal H}
\newcommand{\CN}{\mathcal N}
\newcommand{\CO}{\mathcal O}
\newcommand{\CP}{\mathcal P}
\newcommand{\CR}{\mathcal R}
\newcommand{\ot}{\otimes}
\DeclareMathOperator{\supp}{supp}

\newcommand{\drawRT}[6]{
\pgfmathparse{#2-#1==180};  
\ifnum\pgfmathresult>0 \draw [very thick, shorten >= -0.5, shorten <= -0.5] (#2:#3)--(#1:#3);  
	\draw [very thick, shorten >= -0.6, shorten <= -0.6] (#1:#3) arc (#1:#2:#3); 

\else
	\coordinate (O) at (0,0);
	\coordinate (A) at (#1:#3);
	\coordinate (B) at (#2:#3);
	\coordinate (Ap) at ([shift={({#1-90}:#3)}]A);
	\coordinate (App) at ([shift={({#1+90}:#3)}]A);
	\coordinate (Bp) at ([shift={({#2-90}:#3)}]B);
	\coordinate (Bpp) at ([shift={({#2+90}:#3)}]B);
	\coordinate (X) at (intersection of Ap--App and Bp--Bpp);
	\pgfmathparse{#4==1};
	\ifnum\pgfmathresult>0 \draw [very thick, shorten >= -0.6, shorten <= -0.6] (A) arc (#1:#2:#3); \fi
	\path (X); \pgfgetlastxy{\XCoord}{\YCoord}; 
	\pgfmathsetmacro{\rotateamt}{\ifdim\XCoord<0pt {atan(\YCoord/\XCoord)} \else \ifdim\XCoord>0pt {atan(\YCoord/\XCoord)+180} \else {-\YCoord/abs(\YCoord)*90}\fi \fi}

	\begin{scope}[rotate=\rotateamt]
		\pgfmathparse{#5==1}; 
		\ifnum\pgfmathresult>0 \fill[gray, opacity=0.2, #6]
		let
	 		\p1 = (A),
			\p2 = (B),
	 		\p3 = (X),
			\p4 = (A |- X), 
			\p5 = (B |-X), 
			\n1 = {veclen((\x3-\x1),(\y3-\y1))}, 
			\n2 = {\y1-\y3}, 
			\n3 = {\y2-\y3}  
			in
			([shift=({#1-\rotateamt}:#3)]O) arc ({#1-\rotateamt}:{#2-\rotateamt}:#3) -- ([shift=({asin(\n3/\n1)}:\n1)]X) arc ({asin(\n3/\n1)}:{asin(\n2/\n1)}:\n1);\fi 

		\draw [very thick, shorten >= 0.0, shorten <= 0.0]
			let
	 		\p1 = (A),
			\p2 = (B),
	 		\p3 = (X),
			\p4 = (A |- X), 
			\p5 = (B |-X), 
			\n1 = {veclen((\x3-\x1),(\y3-\y1))}, 
			\n2 = {\y1-\y3}, 
			\n3 = {\y2-\y3}  
			in
			([shift=({asin(\n3/\n1)}:\n1)]X) arc ({asin(\n3/\n1)}:{asin(\n2/\n1)}:\n1); 
	\end{scope}
\fi
}

\begin{document}
\title{Entanglement Wedge Reconstruction via Universal Recovery Channels}
\author{Jordan Cotler}
\author{Patrick Hayden}
\author{Geoffrey Penington}
\author{Grant Salton}
\affiliation{Stanford Institute for Theoretical Physics, Stanford University, Stanford CA 94305, USA}
\author{Brian Swingle}
\affiliation{Stanford Institute for Theoretical Physics, Stanford University, Stanford CA 94305, USA}
\affiliation{Department of Physics, Harvard University, Cambridge MA 02138, USA}
\affiliation{Martin Fisher School of Physics, Brandeis University, Waltham MA 02453, USA}
\author{Michael Walter}
\affiliation{Stanford Institute for Theoretical Physics, Stanford University, Stanford CA 94305, USA}
\affiliation{Korteweg-de Vries Institute for Mathematics, Institute for Theoretical Physics, Institute for Logic, Language and Computation \& QuSoft, University of Amsterdam, The Netherlands}

\begin{abstract}
We apply and extend the theory of universal recovery channels from quantum information theory to address the problem of entanglement wedge reconstruction in AdS/CFT\@.
It has recently been proposed that any low-energy local bulk operators in a CFT boundary region's entanglement wedge can be reconstructed on that boundary region itself.
Existing work arguing for this proposal relies on algebraic consequences of the exact equivalence between bulk and boundary relative entropies, namely the theory of operator algebra quantum error correction.
However, bulk and boundary relative entropies are only approximately equal in bulk effective field theory, and in similar situations it is known that predictions from exact entropic equalities can be qualitatively incorrect.
The framework of universal recovery channels provides a robust demonstration of the entanglement wedge reconstruction conjecture in addition to new physical insights.
Most notably, we find that a bulk operator acting in a given boundary region's entanglement wedge can be expressed as the response of the boundary region's modular Hamiltonian to a perturbation of the bulk state in the direction of the bulk operator.
This formula can be interpreted as a noncommutative version of Bayes' rule that attempts to undo the noise induced by restricting to only a portion of the boundary, and has an integral representation in terms of modular flows.
To reach these conclusions, we extend the theory of universal recovery channels to finite-dimensional operator algebras and demonstrate that recovery channels approximately preserve the multiplicative structure of the operator algebra.
\end{abstract}
\maketitle

The AdS/CFT correspondence is a duality between a gravitational theory in $d+1$-dimensional asymptotically AdS space and a conformal field theory in one fewer spatial dimensions~\cite{Maldacena1999,HKLL,Hamilton2006,Hamilton2007,Kabat2011}.   The CFT lives on the boundary of the bulk AdS space, and the quantum state of the boundary CFT is dual to the state of the quantum gravity theory in the bulk.  Certain CFT states correspond to classical geometries in the bulk, and this important class of states occupies much of our attention.  For such states, there is truly an emergent spatial direction in the bulk theory and understanding how locality with respect to the bulk geometry arises in the boundary theory is a longstanding problem.
One way to approach the problem is to ask which regions of the bulk are completely described by a given region of the boundary.

The question above has been phrased in various forms over the years~\cite{Czech2012,Czech2012n2}, but a natural version is to identify all local bulk operators that can be expressed in terms of boundary operators with support only on the boundary subregion. Significant progress has been made on that problem, starting with the so-called HKLL prescription~\cite{HKLL}.
This method works in certain cases, but falls short in general; there are bulk operators that should be expressible on the boundary subregion that are inaccessible to this technique.  A natural example is the operator corresponding to the area of the minimal surface in the bulk that calculates the entropy of the boundary subregion, the so-called Ryu-Takayanagi surface~\cite{ryu2006holographic,hubeny2007covariant}.
For a single interval in empty AdS, HKLL and related techniques provide the answer, but even for a disconnected subregion of the boundary with total size greater than half, the minimal surface in the bulk falls outside of the purview of HKLL\@. We will review these techniques in detail later.


The problem of finding the dual to a boundary subregion is at the heart of the subject of bulk reconstruction:  given an operator in the bulk, can one find a representation of this operator acting on a subregion of the boundary?  Recently it has been proposed that, for a given subregion $A$ of the conformal boundary, any low-energy bulk operator acting on the \emph{entanglement wedge} of $A$ (defined to be the bulk domain of dependence of
any achronal bulk surface bounded by $A$ and its associated covariant Ryu-Takayanagi minimal surface) can be reconstructed using only information in $A$~\cite{Czech2012,Headrick2014,Wall2014}.
The entanglement wedge reconstruction conjecture was strengthened in~\cite{Pastawski2015} and established in tensor network toy models of holography~\cite{Pastawski2015,hayden2016holographic,nezami2016multipartite}.
Very recently, it was proved in~\cite{Dong2016,JLMS2016} under the condition that the bulk and boundary relative entropies be \emph{exactly} equal.

At large but finite $N$, within the framework of bulk effective field theory, one only expects approximate equality of the bulk and boundary relative entropies.
When similar situations were studied in the quantum information theory literature, it was found that algebraic consequences of exact entropic equalities often do not provide qualitatively correct predictions in the approximate case.
An important and, in fact, closely related example is the exact saturation of strong subadditivity of the von Neumann entropy, which is known to imply that the underlying state is a quantum Markov chain~\cite{hayden2004structure}. The associated algebraic structure slightly generalizes operator algebra quantum error correction, precisely the structure relevant in~\cite{Dong2016}. Near-saturation of strong subadditivity, however, \emph{fails} to imply proximity to a quantum Markov chain state for systems of large Hilbert space dimension~\cite{Ibinson2008,Ding2016}. 

Indeed the most direct generalization of the reconstruction theorem from~\cite{Dong2016} to approximate relative entropy equalities does not lead to full approximate quantum error correction (which is what we establish in this paper), but rather to a significantly weaker condition known as zerobits~\cite{hayden2017approximate,hayden2012weak}. Ignoring this distinction can lead to qualitatively wrong conclusions for code spaces whose dimension grows too fast in the large $N$ limit (for example code spaces containing a large number of black hole microstates)~\cite{hayden2018learning}. For mixed states in such code spaces, the entanglement wedge of a boundary region $A$ may be strictly smaller than the bulk complement of the entanglement wedge of the complementary boundary region $\bar A$ even in the limit $N \to \infty$. A naive application of~\cite{Dong2016} would incorrectly suggest that any operator in the larger latter region can be reconstructed in region $A$; in fact only operators in the entanglement wedge of $A$ can be reconstructed in a state-independent way. 

In this article, we demonstrate the entanglement wedge reconstruction conjecture but without assuming exact equality of bulk and boundary relative entropies.
In doing so, we also provide an \emph{explicit} formula for entanglement wedge reconstruction.
Our analysis builds on recent results in quantum information theory on finding sufficient conditions to approximately reverse the effects of noise.  A quantum channel $\CN$ (\emph{i.e.}, completely-positive, trace-preserving map) is said to be \emph{reversible} if there exists another quantum channel $\CR$ -- known as the \emph{recovery channel} -- such that the composition $\CR\circ\CN$ acts as the identity on all states in the domain of $\CN$ (\emph{i.e.}, $\CR\circ\CN[\rho]=\rho$).
For example, all unitary operations are reversible, with the adjoint of the unitary acting as a recovery channel (since $U^\dagger U = \id$), and quantum error correcting codes are designed around noise processes such that the noise can be reversed on the code subspace.  When a channel $\CN$ is reversible, it has been known for some time how to construct a recovery channel $\CR$~\cite{petz1986sufficient}.
Exact reversibility will almost never be satisfied in practice, however, and for many applications one may only require that a channel be approximately reversible.  For example, in approximate quantum error correction one only requires that the recovered state $\CR\circ\CN[\rho]$ be close to the input $\rho$ up to some small tolerance.  For a channel that is not exactly reversible, it is natural to ask whether or not there exists a recovery channel that works approximately in the above sense~\cite{winter2012stronger}.
This question has spurred a flurry of research and was answered only recently in~\cite{Junge2015}, wherein it was shown that for \emph{any} quantum channel $\CN$, there indeed exists an approximate recovery channel $\CR$ such that $\CR\circ\CN[\rho]\approx\rho,\;\forall\rho$, with the quality of the approximation controlled by the behavior of the relative entropy under the action of the channel $\CN$.  We will review these recent results in the next section.
In the context of AdS/CFT, there is a map from the bulk to the boundary, and a noisy quantum channel arises from tracing over a subregion of the boundary. Our ultimate goal will be to recover from that noise process (\emph{i.e.}, recover from the loss of part of the boundary).

The paper is organized as follows.  We begin with a review of recent results on universal recovery channels, as well as a few basics of AdS/CFT and recent results in holography.  We then apply the theory of universal recovery channels to the problem of entanglement wedge reconstruction in AdS/CFT, and we arrive at an explicit expression for a bulk operator recovered on the boundary.
After discussing its salient structural properties, we sketch how our formula applies to AdS$_3$-Rindler reconstruction, and we conclude with possible avenues for future research.
In the appendix, we prove our entanglement wedge reconstruction result for arbitrary finite-dimensional algebras of observables, thereby obviating simplifying assumptions used in the main body of the paper. To do this, we extend the universal recovery results of~\cite{Junge2015} to finite-dimensional von Neumann algebras.
Importantly, we also prove that approximate recovery channels automatically approximately preserve the multiplicative structure of the original bulk algebra. This ensures that correlation functions of bulk operators can be approximated by correlation functions of boundary reconstructions of the individual operators, even if each operator is reconstructed using a different entanglement wedge.

\section{Preliminaries}
In this section, we review some recent results on recovery channels and AdS/CFT that we later use to establish our results.

\subsection{Universal recovery channels}
To develop intuition, it is helpful to set aside quantum mechanics and consider just probability theory. The problem of reversing the effects of noise then just reduces to the problem of reversing a stochastic map. In particular, given a stochastic map $p(y|x)$ and an observation of $y$, try to infer $x$. One way to do this is to introduce the stochastic map $p(x|y)$ via Bayes' rule
\begin{equation}\label{Bayes}
p(x|y)=\frac{p(x)p(y|x)}{p(y)},
\end{equation}
which has the property that $\sum_y p(x'|y) p(y|x) = \delta_{xx'}$ if the noise can be reversed.
In situations when the noise cannot be perfectly reversed, Bayes' rule provides an excellent (and in many ways optimal) estimate of the input.
Since it will prove useful in solving the quantum version of the problem, it is worth noting that we can trivially rewrite \cref{Bayes} as
\begin{equation}\label{BayesDerivative}
	p(x|y)=\left.\frac{d}{dt}\right|_{t=0} \log\left(\frac{p(y)}{p(x)}+t\; p(y|x)\right).
\end{equation}
In other words, the recovery channel $p(x|y)$ can be expressed as the logarithmic directional derivative of the matrix $p(y){p(x)}^{-1}$ in the direction of the channel $p(y|x)$.

Let us now consider a noncommutative generalization of Bayes' rule.
We would like to find a quantum channel that reverses the action of some input channel.
To make the problem precise, consider two Hilbert spaces $\HS_A$ and $\HS_B$.  Let $S(\HS_A)$ and $S(\HS_B)$ represent the sets of density operators on systems $A$ and $B$, respectively.  A quantum channel $\CN\colon S(\HS_A)\rightarrow S(\HS_B)$ is said to be \emph{reversible} if there exists a quantum channel $\CR\colon S(\HS_B)\rightarrow S(\HS_A)$, called the \emph{recovery channel}, such that
\begin{equation}\label{reversiblechannel}
(\CR\circ\CN)[\rho]=\rho\quad \text{for all }\,\,\rho\in S(\HS_A).
\end{equation}
A simple example in which reversible channels play a starring role is quantum error correction, in which $\CN$ is the composition of encoding some degrees of freedom in a code space $\mathcal H_A$ into a potentially larger Hilbert space, followed by a noise channel (wherein we lose certain degrees of freedom, or otherwise corrupt the encoded state).  $\CR$ is then a decoding map, and \cref{reversiblechannel} corresponds to perfect quantum error correction.

One way to quantify the noisiness of a quantum channel is by comparing the distinguishability of input states and output states using the relative entropy.
Under the action of a quantum channel $\CN$, the relative entropy between two states can never increase.
This fact is known as the \emph{monotonicity of relative entropy} or the \emph{data processing inequality}:
\begin{equation}\label{monotonicityofRE}
D(\rho \|\sigma)\geq D(\CN[\rho]\|\CN[\sigma]),
\end{equation}
where $D(\rho \|\sigma):=\tr \rho\log\rho - \tr\rho\log\sigma$ is the \emph{relative entropy} between $\rho,\sigma$.
If there were a recovery channel $\CR$ such that $(\CR \circ \CN)(\rho)=\rho$ and $(\CR\circ\CN)(\sigma)=\sigma$, monotonicity applied a second time to $\CR$ would imply saturation of \cref{monotonicityofRE}. In fact, the converse is also true.
Equality in \cref{monotonicityofRE} holds if and only if there exists a recovery channel $\CR$ such that $(\CR \circ \CN)(\rho)=\rho$ and $(\CR\circ\CN)(\sigma)=\sigma$.  In this case, Petz~\cite{Ohya2004} identified an exact recovery channel $\CR=\CP_{\sigma,\CN}$ given by
\begin{equation}\label{petzmap}
\CP_{\sigma,\CN}=\sigma^{1/2}\CN^*\bigl[{\CN[\sigma]}^{-1/2}(\cdot){\CN[\sigma]}^{-1/2}\bigr]\sigma^{1/2},
\end{equation}
where $\CN^*$ denotes the adjoint of the channel $\CN$.
Because $\CP_{\sigma,\CN}$ does not depend on $\rho$, it can be used for all $\rho$ saturating~\eqref{monotonicityofRE}
(provided that $\sigma$ is chosen to be full-rank to ensure that the relative entropies in~\eqref{monotonicityofRE} are finite for all states $\rho$).

Failure to saturate~\eqref{monotonicityofRE} means that an exact reversal map $\CR$ cannot exist but there is still the possibility of an approximate reversal map which would behave well in cases of near-saturation, as does Bayes' rule in the stochastic case.   Indeed, an approximate version of the recovery channel was developed by~\citet{Junge2015}, who show that, for any $\rho,\sigma\in S(\HS_A)$ and any quantum channel $\CN$,  there exists a recovery channel $\CR_{\sigma,\CN}$ such that
\begin{equation}\label{JungeMonotonicity}
D(\rho \|\sigma)-D(\CN[\rho] \| \CN[\sigma])\geq -2\log F(\rho,\CR_{\sigma,\CN}\circ\CN[\rho]),
\end{equation}
where $F(\rho,\sigma):=\onenorm{\sqrt{\rho}\sqrt{\sigma}}$ is the fidelity.
The inequality says that the fidelity between the recovered state and the original is controlled by the saturation gap in \cref{monotonicityofRE}, with perfect fidelity in the case of saturation. Importantly, there is no dependence on the dimension of the Hilbert space. 
Moreover, \Citet{Junge2015} gave a concrete expression for the channel $\CR_{\sigma,\CN}$, called the \emph{twirled Petz map} and given by
\begin{equation}\label{twirledPetz}
\CR_{\sigma,\CN}\!:=\!\!\int_\RR \!\!\!dt \, \beta_0(t)\sigma^{-\frac{it}{2}}\CP_{\sigma,\CN}\bigl[{\CN[\sigma]}^{\frac{it}{2}}(\cdot){\CN[\sigma]}^{-\frac{it}{2}}\bigr]\sigma^{\frac{it}{2}},
\end{equation}
where $\CP_{\sigma,\CN}$ is the so-called \emph{Petz map} of \cref{petzmap}, and $\beta_0$ is the probability density $\beta_0(t):=\frac{\pi}{2}{\left(\cosh(\pi t)+1\right)}^{-1}$.
When $\sigma$ is full-rank, both the Petz map and the twirled Petz map are trace-preserving completely positive maps (\emph{i.e.}, quantum channels).

By working with the Choi operator, we can rewrite the recovery channel $\CR_{\sigma,\CN}$ in a form similar to \cref{BayesDerivative}.
For a completely positive map $\mathcal N$, the \emph{Choi operator} is defined by $\Phi_{\mathcal N} := (\operatorname{id} \ot \mathcal N)[\Phi]$, where $\ket{\Phi} = \sum_j \ket{j}\ket{j}$ is an unnormalized maximally entangled state.
In the case of the recovery channel $\CR_{\sigma,\CN}$, the Choi operator can be expressed as
\begin{equation}\label{qBayesDerivative}
\Phi_{\CR_{\sigma,\CN}}= \left.\frac{d}{dt}\right|_{t = 0} \log(\overline{\CN[\sigma]} \otimes \sigma^{-1} + t \, \Phi_{\mathcal{N}^*}),
\end{equation}
where $\overline{\mathcal{N}[\sigma]}$ is the complex conjugate of $\mathcal{N}[\sigma]$, and $\sigma^{-1}$ is the inverse of $\sigma$ on its support.
A proof can be found in the appendix.
This is the appropriate generalization of Bayes' rule to the noncommutative case.
When the channel is reversible, the twirled Petz map $\CR_{\sigma,\CN}$ reduces to the Petz map $\CP_{\sigma,\CN}$.
In the classical case both the Petz map and the twirled Petz map reduce to Bayes' rule.

%

\subsection{AdS/CFT background}
The AdS/CFT correspondence states that quantum gravity in $d+1$ spatial dimensions is dual to a conformal field theory in $d$ spatial dimensions.
There are two main dictionaries that describe the mapping between bulk and boundary quantities in the AdS/CFT correspondence: the \emph{differentiate} dictionary~\cite{Gubser1998,Witten1998}, and the \emph{extrapolate} dictionary~\cite{Banks1998}.
The differentiate dictionary relies on the equivalence between the partition functions of the bulk and boundary theories ($Z_\text{CFT}=Z_\text{grav}$).
On the other hand, the extrapolate dictionary relies on the fact that local CFT operators living in the boundary theory can be expressed as the limit of appropriately weighted bulk fields as they are taken to the conformal boundary of the AdS spacetime.
In particular,
\[ \CO(x)=\lim_{z\to\infty}z^{-\Delta}\phi(x,z), \]
where $\CO$ is a boundary field, $\Delta$ is the scaling dimension of $\CO$, and $\phi$ is a bulk field.  With this equivalence, boundary correlation functions can be expressed as
\[ \expect{\CO(x_1)\dots\CO(x_n)}_\text{CFT}\!=\!\!\lim_{z\to\infty}\!z^{-n\Delta}\!\expect{\phi(x_1,z)\dots\phi(x_n,z)}_\text{bulk}\!. \]

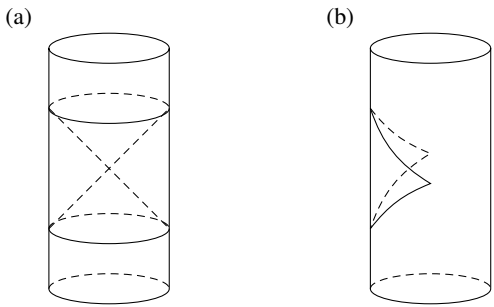
\begin{figure}[!t]
\centering
\begin{tikzpicture}[scale=0.4]
\node (label1) at (-3,8){(a)};
\draw (-2,7) -- (-2,-1) arc (180:360:2cm and 0.5cm) -- (2,7) ++ (-2,0) circle (2cm and 0.5cm);
\draw[densely dashed] (-2,-1) arc (180:0:2cm and 0.5cm);
\draw (-2,1) arc (180:360:2cm and 0.5cm) ;
\draw (-2,5) arc (180:360:2cm and 0.5cm) ;
\draw[densely dashed] (-2,1) arc (180:0:2cm and 0.5cm);
\draw[densely dashed] (-2,5) arc (180:0:2cm and 0.5cm);
\draw[densely dashed] (-2,5) -- (0,3) -- (2,5);
\draw[densely dashed] (-2,1) -- (0,3) -- (2,1);
\end{tikzpicture}
\qquad\qquad\qquad
\begin{tikzpicture}[scale=0.4]
\node (label2) at (-3,8){(b)};
\draw (-2,7) -- (-2,-1) arc (180:360:2cm and 0.5cm) -- (2,7) ++ (-2,0) circle (2cm and 0.5cm);
\draw[densely dashed] (-2,-1) arc (180:0:2cm and 0.5cm);
\draw (-2,5) to [bend right=20] (0, 2.5) to [bend right=15] (-2,1);
\draw[densely dashed] (-2,5) to [bend right=15] (0, 3.5) to [bend right=20] (-2,1);
\end{tikzpicture}
\caption{(a) The HKLL procedure provides a way of writing bulk operators in terms of boundary operators living on a strip in the boundary consisting of all points that are spacelike separated from the bulk point. (b) The causal wedge HKLL procedure provides a way of expressing bulk operators in terms of boundary operators living only in the domain of dependence of a boundary region whose associated causal wedge contains the bulk point.}\label{HKLLandAdSRindler}
\end{figure}

The HKLL procedure~\cite{HKLL} uses the extrapolate dictionary and the bulk equations of motion to write a local bulk field operator $\phi(x,z)$ as a smearing of boundary operators $\CO$ acting on the ``strip'' of boundary points that are spacelike separated from $(x,z)$, as shown in \cref{HKLLandAdSRindler}.  In particular,
\[ \phi(x,z)=\int_\text{strip} dx' K(x,z,x') \mathcal{O}(x'), \]
where the smearing function $K$ can be computed using a mode-sum expansion.  The choice of smearing function $K$ is not unique, and there are different choices of $K$ that will reproduce the same bulk field $\phi$.
For instance, bulk field operators $\phi(x,z)$ can be written using a smearing function that is supported only on a subset of the boundary, namely, the domain of dependence of a boundary subregion whose \emph{causal wedge} contains the bulk point, as shown in \cref{HKLLandAdSRindler}.

Using the HKLL procedure, one can find representations of a given bulk operator on different regions in the boundary.
As observed in~\cite{Almheiri2014}, this redundancy is a reflection of the quantum error correcting properties of AdS/CFT\@.
If we fix a boundary region $A$, the \emph{entanglement wedge reconstruction proposal} asserts that any local bulk operator acting on the entanglement wedge of $A$ has a representation with support only on $A$.

The key input from holography is that bulk and boundary relative entropies are approximately equal~\cite{JLMS2016}.  The relative entropy between two ``nice'' states in the bulk and their associated states on the boundary are equal to leading order in $1/N$.
To be precise, let $\rho$ and $\sigma$ be two bulk states with the same semi-classical geometry, $\rho_a$ and $\sigma_a$ be their reduced density matrices on the entanglement wedge of $A$,
$\widetilde{\rho}$ and $\widetilde{\sigma}$ be the corresponding boundary states, and
$\widetilde{\rho}_A$ and $\widetilde{\sigma}_A$ be their reduced density matrices on region $A$.
\Citet{JLMS2016} then showed that
\begin{equation}\label{JLMSschematic}
D(\widetilde{\rho}_A \| \widetilde{\sigma}_A)=D(\rho_a \| \sigma_a) + \mathcal{O}(1/N).
\end{equation}
At higher orders in $1/N$, the equivalence in \cref{JLMSschematic} is no longer well-defined, since the choice of minimal surface used to define the entanglement wedge becomes state dependent.
Crucially, AdS/CFT provides only a global map from bulk to boundary states ($\rho\mapsto\widetilde\rho$).
Our approach to entanglement wedge reconstruction will be to construct a suitable \emph{local} quantum channel mapping states in the entanglement wedge to states in the boundary region ($\rho_a\mapsto\widetilde\rho_A$).
Only then can we interpret \cref{JLMSschematic} as an approximate saturation of the monotonicity of the relative entropy for a quantum channel, so that \cref{JungeMonotonicity} guarantees the existence of an approximate recovery map.
It is the adjoint of this recovery map that we will ultimately use for reconstruction.


\begin{figure}[!t]
\centering\hspace{2em}
\begin{tikzpicture}
	\draw [use as bounding box] (0,0) circle (1); 
	\node (a) at (-1.5,1.2) {(a)};
	\drawRT{110}{250}{1}{1}{1}{};
	\node at (-0.5,0) {$a$};
	\node at (0.5,0) {$\bar{a}$};
	\node at (-1.2,0) {$A$};
	\node at (1.2,0) {$\bar{A}$};
\end{tikzpicture}
\qquad\qquad\qquad\qquad
\begin{tikzpicture}

	\fill [gray, opacity=0.3] (0,0) circle (1); 
	\draw [use as bounding box] (0,0) circle (1); 
	\node (b) at (-1.5,1.2) {(b)};
	\drawRT{137}{223}{1}{0}{1}{fill opacity=1, color=white};
	\drawRT{-43}{43}{1}{0}{1}{fill opacity=1, color=white};
	\draw [very thick, shorten >= -0.6, shorten <= -0.6] (43:1) arc (43:137:1);
	\draw [very thick, shorten >= -0.6, shorten <= -0.6] (223:1) arc (223:317:1);
	\node at (0,0) {$a$};
	\node at (0,1.2) {$A$};
	\node at (0,-1.2) {$A$};
	\draw [use as bounding box] (0,0) circle (1); 
\end{tikzpicture}
\caption{(a) A bipartition of the boundary into a connected piece $A$ and its complement $\bar{A}$.  In this case, the entanglement wedge of $A$ coincides with the causal wedge of $A$, labeled by $a$ in the figure.  $\bar{a}$ is then the complement of $a$.  (b) A bipartition of the boundary into $A$ and $\bar{A}$ such that $A$ consists of two disconnected components.  In the figure above, $A$ spans just more than half of the boundary, and in this case the entanglement wedge of $A$ is not simply the union of the causal wedges of each piece of $A$.  In the bulk, $a$ represents the entanglement wedge of $A$ and $\bar{a}$ is the complement of $a$.}\label{bulkboundarydecomp}
\end{figure}
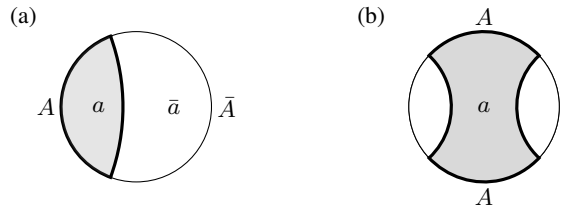

\section{Entanglement wedge reconstruction}
We will make the following assumptions:
\begin{enumerate}
\item the bulk Hilbert space contains a code subspace that is mapped via a quantum channel into the CFT Hilbert space, and
\item the JLMS relative entropy condition holds to leading order in $1/N$.
\end{enumerate}
In fact, the first condition is slightly relaxed in the most general version of our result (\cref{thm:algebraic entanglement wedge}).
Just for the purposes of illustration, we will also pretend that the bulk and boundary Hilbert spaces admit simple tensor factorizations, since this is a familiar convention in the community.  However, the reader is cautioned that this convention is not actually correct -- even for free theories the Hilbert space does not factorize, and the problem is only compounded for gauge theories.  The proper approach is to work at the level of algebras of observables, without ever making a tensor product assumption.  As such, we adopt the more general algebraic approach in the appendix, proving all claims made in this paper rigorously at the level of finite-dimensional von Neumann algebras.


With these assumptions, let us formalize the problem.
Let $\HS_\text{code}$ be a code space with density operators $S(\HS_\text{code})$, and let $\HS_\text{CFT}$ be the Hilbert space of a CFT with density operators $S(\HS_\text{CFT})$.
There are many valid definitions of what it means to be a code space in AdS/CFT, but our results hold for any suitable definition.
For instance, one can define the code space to be the set of all states formed by acting with a finite number of low-energy, local bulk operators on the vacuum~\cite{Almheiri2014,Dong2016}.
In this context, low-energy means the action of the operator does not change the bulk geometry appreciably.
The AdS/CFT correspondence relates states in $S(\HS_\text{code})$ to states in $S(\HS_\text{CFT})$.
We model this relationship by an isometry $J\colon\HS_\text{code}\to \HS_\text{CFT}$ embedding the code space into the CFT Hilbert space.
Note that $\HS_\text{code}$ can be identified with its image under the isometry $J$, resulting in the \emph{code subspace} of previous works~\cite{Dong2016,Harlow2016}.
In general, one could consider an arbitrary quantum channel mapping states on the code space to states on the CFT Hilbert space. We prove our general result in \cref{thm:algebraic entanglement wedge} without requiring the mapping between code and CFT states to be an isometry.

We now partition the CFT into two regions $A$ and $\bar{A}$, and we partition the bulk into $a$ and $\bar{a}$, where $a$ is supported only on the entanglement wedge of $A$, as shown in \cref{bulkboundarydecomp}.  As discussed earlier, we assume that $\HS_\text{CFT}=\HS_A\ox\HS_{\bar{A}}$ and that $\HS_\text{code}=\HS_a\ox\HS_{\bar{a}}$.
The problem of entanglement wedge reconstruction then amounts to constructing a boundary observable $\CO_A$ supported only on $A$, such that, for any bulk operator $\phi_a$ supported in the entanglement wedge $a$ of $A$
\begin{equation}\label{bulkboundaryexpectationvalues}
\big\lvert\langle \CO_A\rangle_{J \rho J^\dagger} - \langle \phi_a \rangle_\rho\bigr\rvert \leq \delta \lVert\phi_a\rVert,
\end{equation}
for all $\rho\in S(\HS_\text{code})$ and for some small $\delta>0$.
Using our new notation, \cref{JLMSschematic} says that, for all $\rho,\sigma\in S(\HS_\text{code})$,
\begin{equation}\label{JLMSsubsystem}
\left| D(\rho_a \Vert \sigma_a)-D\bigl({(J\rho J^\dagger)}_A \Vert {(J\sigma J^\dagger)}_A\bigr) \right| \leq\epsilon,
\end{equation}
where $\epsilon$ is controlled by $1/N$, and the notation ${(\cdot)}_A:=\tr_{\bar A}(\cdot)$ is a shorthand we will use throughout.

Note that in this paper we are using the approximate equality of relative entropies in $a$ and $A$ in order to prove that operators in region $a$ can be reconstructed in $A$. In contrast, the starting assumption in \cite{Dong2016} was an exact equality between relative entropies in $\bar a$ and $\bar A$. As discussed in the introduction, an approximate version of this latter condition implies only that the zerobits of region $a$ are encoded in region $A$, which is a significantly weaker condition than full entanglement wedge reconstruction \cite{hayden2017approximate,hayden2012weak,hayden2018learning}.

Since the relative entropies approximately agree, one might expect that we can find a universal recovery channel that would undo the effect of the partial trace over $\bar{A}$.  However, there is an obstacle:
\emph{a priori}, ${(J\rho J^\dagger)}_A$ depends on the state $\rho$ defined on the whole bulk, not just on the reduced state on the entanglement wedge, $\rho_a$.  In this form, the theory of recovery channels is not applicable.  To overcome this challenge, we will first restrict the recovery problem to special code states of the form $\rho = \rho_{a} \otimes \sigma_{\bar a}$, where $\sigma_{\bar a}$ is some fixed fiducial state.
We thus obtain a quantum channel $\rho_a \mapsto {(J\rho J^\dagger)}_A$ mapping states on the entanglement wedge to states on the boundary region. We will then verify that the recovery map $\CR$ obtained for this channel works in fact for \textit{all} code states $\rho$, only increasing the error by a small amount, since \cref{JLMSsubsystem} also implies that that the CFT states corresponding to any $\rho$ and its factorized version $\rho_a\ox\sigma_{\bar a}$ are approximately indistinguishable on the boundary region $A$.
The adjoint $\mathcal{R}^*$ of the recovery channel will then map bulk operators $\phi_a$ supported in the entanglement wedge to boundary operators $\CO_A$ supported on $A$ and satisfying \cref{bulkboundaryexpectationvalues}, thereby achieving entanglement wedge reconstruction.

In more mathematical detail, we first define the \emph{local channel} $\CN\colon S(\HS_a)\to S(\HS_A)$ by
\begin{equation}\label{eq:the map}
\CN[\rho_a]:=\tr_{\bar A}\bigl[J(\rho_a\ox\sigma_{\bar a})J^\dagger\bigr] ={\bigl( J(\rho_a\ox\sigma_{\bar a})J^\dagger \bigr)}_A,
\end{equation}
for all states $\rho_a\in S(\HS_a)$, where $\sigma_{\bar a}$ is some fixed full-rank state.
If we also choose a full-rank $\sigma_a\in S(\HS_a)$ then we can use \cref{twirledPetz} to obtain a recovery channel $\CR=\CR_{\sigma_a,\CN}$ such that, for all $\rho_a\in S(\HS_a)$,
\begin{equation*}
-2\log F(\rho_a,\CR\circ\CN[\rho_a])\leq\big\lvert D(\rho_a \| \sigma_a) - D(\CN[\rho_a] \| \CN[\sigma_a]) \big\rvert.
\end{equation*}
However, by \cref{JLMSsubsystem}, we have
\begin{equation*}
\big\lvert D(\rho_a \| \sigma_a) - D(\CN[\rho_a] \| \CN[\sigma_a]) \big\rvert \leq\epsilon,
\end{equation*}
and therefore we conclude that the recovery channel $\CR$ works with high fidelity (cf.~\cite[Corollary 6.1]{Junge2015}).
By one of the Fuchs-van de Graaf inequalities~\cite{FuchsvandeGraaf}, this implies that
\begin{equation*}
\big\lVert \rho_a - \CR[\CN[\rho_a]] \big\rVert_1 \leq2\sqrt{\epsilon}:=\delta_1
\end{equation*}
for all $\rho_{a} \in S(\HS_a)$.
We now show that the channel $\CR$ recovers the reduced state on the entanglement wedge for arbitrary code states $\rho$, not just for those of the form $\rho = \rho_a \ot \sigma_{\bar a}$:
\begin{align*}
	&\onenorm{\CN[\rho_a] - {(J\rho J^\dagger)}_A}^2 =
	\onenorm{{\bigl(J(\rho_a\ox\sigma_{\bar{a}})J^\dagger\bigr)}_A - {(J\rho J^\dagger)}_A}^2\\
	&\leq 2 \ln 2 \, D\Bigl({\bigl(J (\rho_a\ox\sigma_{\bar{a}}) J^\dagger\bigr)}_A \Vert\, {(J \rho J^\dagger)}_A\Bigr)\\
	&\leq (2 \ln 2) \epsilon =: \delta_2^2,
\end{align*}
where the first inequality is Pinsker's inequality and the second inequality is \cref{JLMSsubsystem}, with one state set to $\rho$ and the other set to $\rho_a\ox\sigma_{\bar{a}}$.
Therefore, we obtain that, for all $\rho\in S(\HS_\text{code})$,
\begin{align}
&\quad\;\onenorm{\rho_a - \CR[{(J\rho J^\dagger)}_A]}\notag\\
&\leq \onenorm{\rho_a - \CR[\CN[\rho_a]]} + \onenorm{\CR[\CN[\rho_a]] - \CR[{(J\rho J^\dagger)}_A]} \notag \\
&\leq \onenorm{\rho_a - \CR[\CN[\rho_a]]} + \onenorm{\CN[\rho_a] - {(J\rho J^\dagger)}_A } \notag \\
&\leq \delta_1 + \delta_2 =: \delta.\label{RrecoversJ}
\end{align}
Thus, we conclude that $\CR$ recovers arbitrary bulk states in the entanglement wedge with high fidelity, as desired.
We now show that the adjoint of the map $\CR$ solves the entanglement wedge reconstruction problem in the form of \cref{bulkboundaryexpectationvalues}.
Given a bulk operator $\phi_a$ supported in entanglement wedge $a$ of $A$, define $\CO_A=\mathcal R^*[\phi_a]$.
Then we have that, for all $\rho\in S(\HS_\text{code})$,
\begin{align*}
&\quad\; \bigl\lvert \expect{\CO_A}_{J\rho J^\dagger} - \expect{\phi_a}_\rho \bigr\rvert \\
&= \bigl\lvert \tr \mathcal R^*[\phi_a] {(J\rho J^\dagger)}_A - \tr \phi_a \rho_a \bigr\rvert \\
&= \bigl\lvert \tr \phi_a \mathcal R[{(J\rho J^\dagger)}_A] - \tr \phi_a \rho_a \bigr\rvert \\
&= \bigl\lvert \tr \phi_a (\mathcal R[{(J\rho J^\dagger)}_A] - \rho_a) \bigr\rvert \\
&\leq \bigl\lVert \mathcal R[{(J\rho J^\dagger)}_A] - \rho_a \bigr\rVert_1 \, \bigl\lVert \phi_a \bigr\rVert \leq \delta \bigl\lVert \phi_a\bigr\rVert,
\end{align*}
where the first inequality is H\"older's inequality, and the second is \cref{RrecoversJ}.
This is the desired approximate equality of one-point functions.

\section{Correlation functions}
Given a set of $n$ bulk operators $\{ \phi^{(i)}_a \}$ acting on the entanglement wedge~$a$ of~$A$, our result implies that we can calculate their $n$-point correlation function as the expectation value of the boundary operator $O_A = \CR^*[\prod_i \phi^{(i)}_a]$ obtained by reconstructing the composite bulk operator $\prod_i \phi^{(i)}_a$.
However, one might hope that the reconstructed operators approximately reproduce the bulk algebra in the sense that
\begin{equation}\label{RecoveryHomo}
\Expect{\prod_i \phi_a^{(i)}}_\rho\approx \Expect{\prod_i \CO_A^{(i)}}_{J\rho J^\dagger},
\end{equation}
where $\CO_A^{(i)} = \mathcal R^*(\phi_a^{(i)})$.
When the bulk and boundary relative entropies are exactly equal, it is known that $\mathcal{R}^*$ is an algebra homomorphism~\cite[Proposition~8.4]{Ohya2004}, so that~\cref{RecoveryHomo} holds with equality.
We prove in \cref{thm:algebraic entanglement wedge} in the appendix that when the bulk and boundary relative entropies are only approximately equal, as in \cref{JLMSsubsystem}, then \cref{RecoveryHomo} still holds approximately, although the size of the error may grow quadratically with~$n$.

Furthermore, we show in \cref{corr:multientwedge} that this continues to be true even when different operators are reconstructed using the entanglement wedges~$a_i$ of different boundary regions~$A_i$:
\begin{equation}\label{multientwedgecorrelation}
\Expect{\prod_i \phi_{a_i}^{(i)}}_\rho\approx \Expect{\prod_i \CO_{A_i}^{(i)}}_{J\rho J^\dagger},
\end{equation}
where $\CO_{A_i}^{(i)} = \mathcal R_{A_i}^*[\phi_{a_i}^{(i)}]$ and $\mathcal R_{A_i}$ is the recovery map for boundary region~$A_i$.  This is illustrated in \cref{multientwedgefigure}.
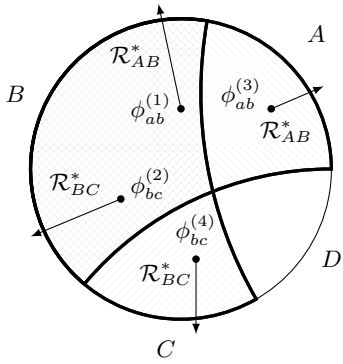
\begin{figure}[!t]
\centering
\begin{tikzpicture}
	\drawRT{0}{230}{2}{1}{1}{pattern=north west lines};
	\drawRT{80}{300}{2}{1}{1}{pattern=north east lines};
	\drawRT{0}{230}{2}{1}{0}{};
	\drawRT{80}{300}{2}{1}{0}{};
	\fill (0,0.8) circle (0.05);
	\draw (-0.4,0.8) node{$\phi^{(1)}_{ab}$};
	\fill (-0.8,-0.4) circle (0.05);
	\draw (-0.4,-0.2) node {$\phi^{(2)}_{bc}$};
	\fill (1.2,0.8) circle (0.05);
	\draw (0.8,1) node {$\phi^{(3)}_{ab}$};
	\fill (0.2,-1.2) circle (0.05);
	\draw (0.2,-0.8) node {$\phi^{(4)}_{bc}$};
	\draw[-latex] (0,0.8) -- (-0.3,2.2);
	\draw[-latex] (-0.8,-0.4) -- (-2,-0.9);
	\draw[-latex] (1.2,0.8) -- (1.9,1.1);
	\draw[-latex] (0.2,-1.2) -- (0.2,-2.2);
	\draw (-0.6,1.5) node {$\mathcal{R}^*_{AB}$};
	\draw (1.4,0.5) node {$\mathcal{R}^*_{AB}$};
	\draw (-1.4,-0.2) node {$\mathcal{R}^*_{BC}$};
	\draw (-0.2,-1.4) node {$\mathcal{R}^*_{BC}$};
	\draw (1.8,1.8) node {$A$};
	\draw (-2.2,1) node {$B$};
	\draw (-0.2,-2.4) node {$C$};
	\draw (2,-1.2) node {$D$};

	\draw (0,0) circle (2); 
	\pgfresetboundingbox
	\path[use as bounding box] (-2.5,-2.7) rectangle (2.3,2.3);
\end{tikzpicture}
\caption{We show in \cref{multientwedgecorrelation} that correlation functions of bulk operators can be computed by pushing each bulk operator to the boundary separately, and computing the expectation value in the boundary theory.  The bulk operators need not live in the same entanglement wedge.  In this figure, the boundary is decomposed into four regions: $A$, $B$, $C$, and $D$.  Regions $AB$ and $BC$ have a non-trivial intersection, and the bulk operators $\phi$ are localized to the regions as shown.  We can use the recovery maps $\mathcal{R}_{AB}^*$ and $\mathcal{R}_{BC}^*$ to push the operators to $AB$ and $BC$, respectively.}\label{multientwedgefigure}
\end{figure}

\section{An explicit formula}

Our ultimate goal is an explicit formula for approximate entanglement wedge reconstruction.
Thus we would like to calculate $\CO_A=\mathcal R^*[\phi_a]$ explicitly by using \cref{twirledPetz}.
Recall that the recovery channel $\CR$ depends on our choice of $\sigma_a$ and, through the channel $\CN$ from \cref{eq:the map}, also on the choice of $\sigma_{\bar a}$.
The result is particularly satisfying when both $\sigma_a$ and $\sigma_{\bar a}$ are chosen to be maximally mixed.
It is important to emphasize that this is just a convenient choice that we make in order to simplify our expressions.
For an infinite-dimensional code space, such a choice would not be well-defined, and instead another choice of average code state could be used.
With this simplification, we find that, for all bulk operators $\phi_a$ with support in the entanglement wedge $a$,
\begin{align}\label{ExplicitRecoveryHilbert}
&\qquad \CO_A := \CR^*[\phi_a]=\\
&\frac{1}{d_\text{code}}\!\int_{\RR}\!\!\!dt\, \beta_0(t) e^{\frac{1}{2}(1-it)H_A}
\tr_{\bar{A}}\left[J(\phi_a \ox \id_{\bar{a}}) J^\dagger\right]e^{\frac{1}{2}(1+it)H_A},\notag
\end{align}
where $H_A=-\log {(J\tau J^\dagger)}_A$ is the boundary modular Hamiltonian on subregion $A$ associated with the maximally mixed state $\tau$ on the code subspace.
The above expression makes it clear that the natural basis one should use for entanglement wedge reconstruction is the eigenbasis of the modular Hamiltonian.
What is more, the recovery channel can be expressed in the form of a logarithmic directional derivative, as in \cref{qBayesDerivative}:
\begin{equation}\label{LogDerivExplicit}
\CO_A = \mathcal R^*[\phi_a]
= - \frac{1}{d_{\text{code}}} \left. \frac{d}{dt}\right|_{t=0} H_A[\tau_{\text{code}} + t \, \phi_a \otimes \id_{\bar a}]
\end{equation}
where we write $H_A[\rho] := -\log {(J \rho J^\dagger)}_A$ for the boundary modular Hamiltonian on subregion $A$ associated with a bulk state $\rho$.
In other words, the boundary operator corresponding to $\phi_a$ can be computed as the response in the boundary modular Hamiltonian $H_A$ to a perturbation of the maximally mixed code state in the direction of the operator $\phi_a$.

\Cref{ExplicitRecoveryHilbert,LogDerivExplicit} are explicit expressions for entanglement wedge reconstruction that are meaningful even when the bulk and boundary relative entropies are not exactly equal.
When the relative entropies are exactly equal for all $\rho$, they reduce to the simple Petz map, which is equivalent to existing notions of operator algebra quantum error correction, as applied to bulk reconstruction.

In \cref{ExplicitRecoveryHilbert}, we first map our bulk operator $\phi_a$ to the \textit{entire} boundary via $\phi_a \mapsto J (\phi_a\ot\id_{\bar a}) J^\dagger$.
One might wonder what the connection is between this mapping and the global HKLL reconstruction procedure~\cite{HKLL}, discussed in the preliminaries, which likewise produces an operator $\CO_{\text{HKLL}}$ supported on the full CFT Hilbert space and satisfying $\expect{\phi_a}_\rho = \expect{\CO_{\text{HKLL}}}_{J\rho J^\dagger}$.
The latter conditions means that $J^\dagger \CO_{\text{HKLL}} J = \phi_a \ot \id_{\bar a}$. Hence:
\begin{equation*}
	J (\phi_a \ot \id_{\bar a}) J^\dagger = JJ^\dagger \CO_{\text{HKLL}} JJ^\dagger.
\end{equation*}
The map $JJ^\dagger$ is precisely the projection onto the code subspace.
This means that, in order to compute the term $J (\phi_a \ot \id_{\bar a}) J^\dagger$ in our formula~\eqref{ExplicitRecoveryHilbert}, we can leverage the global HKLL procedure to map out the bulk operator to a global boundary operator, if we make sure to subsequently project the result onto the code subspace.
Such a projection onto the code subspace is not necessarily complicated.

In the appendix, we work through an explicit calculation involving our reconstruction formula.  The example is analogous to Rindler wedge reconstruction, although it is only strictly true for free fields.
We consider a bulk operator $\phi_a$ in AdS$_3$, localized in the entanglement wedge of a boundary interval, and we choose a two-dimensional code space spanned by states $\ket{\widetilde 0}$ and $\phi_a \ket{\widetilde 0}$, where $\ket{\widetilde 0}$ is the vacuum state.
Using only our recovery formula and global HKLL, we find an expression for reconstructed bulk operators as a mode expansion supported on the Rindler wedge.

\section{Discussion}

The mapping of bulk operators to boundary subregions was recognized as a problem in operator algebra quantum error correction in~\cite{Almheiri2014}.
This important conceptual advance, along with the insight that bulk and boundary relative entropies agree~\cite{JLMS2016}, imply that any low-energy bulk operator in the entanglement wedge of a boundary region should be representable as an operator acting on that boundary region~\cite{Dong2016,JLMS2016}. In this sense, a boundary region is dual to its entanglement wedge.

In this article, we used recent advances in quantum information theory to provide a robust demonstration of this result that does not assume that the bulk and boundary relative entropies are exactly equal.
In addition, we found a satisfyingly simple, explicit formula for the boundary operator: namely, that it can be computed as the response of the boundary modular Hamiltonian of the subregion to a perturbation of the average code state in the direction of the bulk operator.

Our argument did not rely on the structural consequences implied by exact equality of relative entropies, assumed in the proof of the entanglement wedge reconstruction conjecture in~\cite{Dong2016}.
That said, the argument presented here still assumes the finite dimensionality of the associated von Neumann algebras.
Most of our expressions can be applied formally even in the infinite-dimensional setting, but it would a worthwhile project to try to rigorously extend our results to the infinite-dimensional case.
It seems likely that additional hypotheses will be required in order to ensure the existence of the local channel that was instrumental in our argument.

In independent work~\cite{Faulkner2017}, Faulkner and Lewkowycz have arrived at a formula for entanglement wedge reconstruction which also involves modular flow. Their approach builds on the insights of~\cite{JLMS2016} and uses the free field physics of the bulk to argue that entanglement wedge reconstruction involves integrating the modular flow against a certain (generally unknown) kernel. It would interesting to understand their results from the perspective of our framework; most likely this requires further exploring the consequences of the free field assumption in the bulk which we do not a priori need to assume in our approach. 

\smallskip

\begin{acknowledgments}
\textbf{Acknowledgments.} We thank Mario Berta, Tom Faulkner, Daniel Harlow, Eliot Hijano, Aitor Lewkowycz, Sepehr Nezami, Jonathan Oppenheim, David Sutter, Mark Van Raamsdonk, and Mark Wilde for helpful discussions and feedback.
BGS is supported by the Simons Foundation's It From Qubit collaboration; through a Simons Investigator Award to Senthil Todadri; and by MURI grant W911NF-14-1-0003 from ARO\@.
GS is supported by an NSERC postgraduate scholarship.
JC is supported by the Fannie and John Hertz Foundation and the Stanford Grad Fellowship program.
MW and PH gratefully acknowledge support from the Simons Foundation's Investigator program and It from Qubit collaboration, as well as AFOSR grant FA9550-16-1-0082.
PH is also supported by CIFAR\@.
MW also acknowledges financial support by the NWO through Veni grant no.~680-47-459.
\end{acknowledgments}
\bibliography{BulkRecovery}
\cleardoublepage\onecolumngrid%
\section{Appendix}
\appendix

\section{Finite-Dimensional Von Neumann Algebras}
A finite-dimensional \emph{von Neumann algebra} is a (unital) subalgebra $\mathcal A\subseteq B(\mathcal H)$ of the linear operators on some finite-dimensional Hilbert space $\mathcal H$.
We define the set of \emph{states} as the intersection $S(\mathcal A) = \mathcal A \cap S(\mathcal H)$ of the algebra with the set of all density operators on the Hilbert space $\mathcal H$, denoted by $S(\mathcal H)$. We denote the larger space of positive linear functionals (with no normalization condition) on $\mathcal A$ by $P(\mathcal A)$.
We write $\expect{\phi}_\rho := \tr \rho\phi$ for the expectation value of an operator $\phi\in\mathcal A$ in state $\rho\in S(\mathcal A)$.
(To any state $\rho\in S(\mathcal A)$ we may assign the positive normalized linear functional $\phi \mapsto \expect\phi_\rho$, thereby connecting $S(\mathcal A)$ with the standard definition of states on a von Neumann algebra.)

In this way, we may lift standard definitions in quantum information theory to finite-dimensional von Neumann algebras.
For example, the \emph{trace norm difference} $\lVert\rho-\sigma\rVert_1$, the \emph{relative entropy} $D(\rho\Vert\sigma)$ and the \emph{fidelity} $F(\rho,\sigma)$ for $\rho,\sigma\in S(\mathcal A)$ can be defined in the usual way by $\lVert\rho-\sigma\rVert_1 := \tr\lvert\rho-\sigma\rvert$, $D(\rho\Vert\sigma):=\tr[\rho\log\rho-\rho\log\sigma]$, and $F(\rho,\sigma):=\lVert\sqrt\rho\sqrt\sigma\rVert_1$, respectively, in agreement with their abstract definitions for von Neumann algebras.
For a completely positive map $\mathcal N\colon P(\mathcal A)\to P(\mathcal B)$, the \emph{adjoint} or \emph{dual} channel $\mathcal N^*\colon \mathcal B\to\mathcal A$ is defined by demanding that $\expect{\phi}_{\mathcal N[\rho]} = \expect{\mathcal N^*[\phi]}_\rho$ for all $\rho\in S(\mathcal A)$ and $\phi\in\mathcal B$.
If $\mathcal N$ is trace-preserving then it is called a \emph{quantum channel} (equivalently, the dual map is unital, \emph{i.e.}, $\mathcal N^*[\id]=\id$).
It is called a \emph{quantum operation} if $\mathcal N$ is merely trace non-increasing (equivalently, $\mathcal N^*[\id]\leq\id$).

Consider a (unital) subalgebra $\mathcal A\subseteq \mathcal B$.
For any state $\rho\in S(\mathcal B)$, we define its \emph{restriction} $\rho\vert_{\mathcal A}$ as the unique element in $S(\mathcal A)$ such that $\expect{\phi}_{\rho\vert_{\mathcal A}} = \expect\phi_\rho$ for all $\phi\in\mathcal A$; the assignment $\rho\mapsto\rho\vert_{\mathcal A}$ defines a quantum channel.
The inclusion map $\mathcal E_{\mathcal A}\colon S(\mathcal A)\subseteq S(\mathcal B)$ is likewise a quantum channel, sometimes referred to as a \emph{state extension}.
(In the language of von Neumann algebras, it is the predual of a conditional expectation onto $\mathcal A$.)
Importantly,
\begin{equation}\label{eq:state extension}
	\mathcal E_{\mathcal A}[\rho]\big\vert_{\mathcal A}=\rho
\end{equation}
for all $\rho\in S(\mathcal A)$.
Since $\mathcal E_{\mathcal A}$ is just the inclusion map, it is immediate that
\begin{equation}\label{eq:same relative entropies and fidelities}
  D(\rho\Vert\sigma)=D(\mathcal E_{\mathcal A}[\rho]\Vert\mathcal E_{\mathcal A}[\sigma])
  \quad\text{and}\quad
  F(\rho, \sigma)=F(\mathcal E_{\mathcal A}[\rho], \mathcal E_{\mathcal A}[\sigma])
\end{equation}
for all $\rho,\sigma\in S(\mathcal A)$.

Lastly, we show that there is a natural generalization of Stinespring's dilation theorem to quantum operations on finite-dimensional von Neumann algebras.
\begin{lem}\label{lem:stinespring}
Let $\mathcal N\colon P(\mathcal A)\to P(\mathcal B)$ be a quantum operation on finite-dimensional von Neumann algebras with $\mathcal A \subseteq B(\mathcal H_A)$ and $\mathcal B \subseteq B(\mathcal H_B)$. Then:
\begin{align}
\mathcal N (\rho) = [V \mathcal E_A (\rho) V^\dagger] \big\vert_{\mathcal B}
\end{align}
where $V: \mathcal H_A \to \mathcal H_B \otimes \mathcal H_E$ (for some auxiliary Hilbert space $\mathcal H_E$) and $V^\dagger V \leq \id$.
\end{lem}
\begin{proof}
We define $$\tilde{\mathcal N}: P(\mathcal H_A) \to P(\mathcal H_B), \,\,\,\,\omega \to \mathcal E_{\mathcal B} [\mathcal N[\omega\big\vert_{\mathcal A}]].$$ This is a quantum operation between states on Hilbert spaces and hence has a Stinespring dilation $V: \mathcal H_A \to \mathcal H_B \otimes \mathcal H_E$ such that $\tilde{\mathcal N} = \tr_E V(\cdot) V^\dagger$. However, it follows from \eqref{eq:state extension} that
\begin{align}
\mathcal N = [\tilde{\mathcal N} \circ \mathcal E_{\mathcal A} (\cdot)] \big\vert_{\mathcal B},
\end{align}
and thus we obtain the lemma.
\end{proof}
We refer to~\cite{Ohya2004,petz2007quantum,Harlow2016} for more detailed expositions of the theory of finite-dimensional von Neumann algebras.
\section{Approximate Recovery Maps for Von Neumann Algebras}

We now extend the universal recovery result of~\cite{Junge2015} to finite-dimensional von Neumann algebras (c.f., \cref{JungeMonotonicity} herein).

\begin{lem}\label{lem:universal recovery}
Let $\mathcal N\colon S(\mathcal A)\to S(\mathcal B)$ be a quantum channel of finite-dimensional von Neumann algebras and $\rho,\sigma\in S(\mathcal A)$ states such that $\supp\rho\subseteq\supp\sigma$. Then,
\[ D(\rho\Vert\sigma) - D(\mathcal N[\rho]\Vert\mathcal N[\sigma]) \geq -2\log F(\rho, (\mathcal R_{\sigma,\mathcal N} \circ \mathcal N)[\rho]), \]
where
\begin{equation}\label{eq:universal recovery map vN}
  \mathcal R_{\sigma,\mathcal N}[\gamma]
:= \int dt \, \beta_0(t) \, \sigma^{-it/2} \mathcal P_{\sigma,\mathcal N}\bigl[{\mathcal N[\sigma]}^{it/2} \,\gamma\, {\mathcal N[\sigma]}^{-it/2}\bigr] \sigma^{it/2}
\end{equation}
is a quantum operation defined in terms of the Petz recovery map
\[ \mathcal P_{\sigma,\mathcal N}[\gamma] = \sigma^{1/2} \mathcal N^*\bigl[{\mathcal N[\sigma]}^{-1/2} \,\gamma\, {\mathcal N[\sigma]}^{-1/2}\bigr] \sigma^{1/2} \]
and the probability distribution $\beta_0(t) = \frac\pi2{(\cosh(\pi t)+1)}^{-1}$.
\end{lem}
\begin{proof}
  By assumption, $\mathcal A\subseteq B(\mathcal H_A)$ and $\mathcal B\subseteq B(\mathcal H_B)$ for finite-dimensional Hilbert spaces $\mathcal H_A,\mathcal H_B$.
  We denote by $\mathcal E_{\mathcal A}$ and $\mathcal E_{\mathcal B}$ the corresponding state extension maps, defined above.
  As in \cref{lem:stinespring}, we now consider
  \[ \widetilde{\mathcal N}\colon S(\mathcal H_A)\to S(\mathcal H_B), \quad \omega \mapsto \mathcal E_{\mathcal B}[\mathcal N[\omega\rvert_{\mathcal A}]]. \]
  This is a quantum channel between density operators on Hilbert spaces and hence~\cite[Theorem 2.1]{Junge2015} is applicable.
  It states that
	\[ D(\omega\Vert\chi) - D(\widetilde{\mathcal N}[\omega]\Vert\widetilde{\mathcal N}[\chi]) \geq -2\log F(\omega, (\mathcal R_{\chi,\widetilde{\mathcal N}} \circ \widetilde{\mathcal N})[\omega]) \]
	for any pair of density operators $\omega,\chi\in S(\mathcal H)$ such that $\supp\omega\subseteq\supp\chi$.
	We now make the choice $\omega=\mathcal E_A[\rho]$ and $\chi=\mathcal E_A[\sigma]$.
	Then, using \cref{eq:state extension,eq:same relative entropies and fidelities}, it follows that
	$D(\omega\Vert\chi) = D(\rho\Vert\sigma)$,
	$D(\widetilde{\mathcal N}[\omega]\Vert\widetilde{\mathcal N}[\chi]) = D(\mathcal N[\rho]\Vert\mathcal N[\sigma])$,
	$\mathcal R_{\chi,\widetilde{\mathcal N}} \circ \mathcal E_{\mathcal B} = \mathcal E_A \circ \mathcal R_{\sigma,\mathcal N}$, and hence
	$F(\omega, (\mathcal R_{\chi,\widetilde{\mathcal N}} \circ \widetilde{\mathcal N})[\omega]) = F(\rho, (\mathcal R_{\sigma,\mathcal N} \circ \mathcal N)[\rho])$.
	The lemma follows immediately.
\end{proof}

If a quantum channel $\mathcal N\colon S(\mathcal A)\to S(\mathcal B)$ is exactly reversible then it can be reversed by the Petz recovery map $\mathcal P=\mathcal P_{\sigma,\mathcal N}$ for any faithful state $\sigma$.
In this case, $\mathcal P^*$ is \emph{multiplicative} in the sense that $\mathcal P^*[\phi' \phi]=\mathcal P^*[\phi'] \mathcal P^*[\phi]$ (e.g.,~\cite[Proposition~8.4]{Ohya2004}).
If $\mathcal R$ is an arbitrary quantum operation that reverses $\mathcal N$ (\emph{i.e.}, $(\mathcal R \circ \mathcal N)[\rho]=\rho$ for all $\rho$) then it is still true that
$\Expect{\mathcal R^*[\phi_1] \mathcal R^*[\phi_2] \dots}_{\mathcal N[\rho]} = \Expect{\mathcal R^*[\phi_1 \phi_2 \dots]}_{\mathcal N[\rho]} = \Expect{\phi_1 \phi_2 \dots}_{\rho}$
for all $\rho\in S(\mathcal A)$, or, equivalently, that $\mathcal N^*[\mathcal R^*[\phi_1] \mathcal R^*[\phi_2] \dots] = \mathcal N^*[\mathcal R^*[\phi_1 \phi_2 \dots]]$.
In fact this is true even when a different (exact) recovery map is used for each operator $\phi_i$. We now prove that approximate reversibility implies approximate multiplicativity. Correlation functions are reconstructed up to an error that grows at most quadratically with $n$, even when different recovery maps $\mathcal R_i$, correcting different subalgebras $\mathcal A_i$ of the original algebra $\mathcal A$, are used for each operator $\phi_i$.

\begin{thm}\label{thm:approx mul}
  Let $\mathcal N\colon P(\mathcal A)\to P(\mathcal B)$ and $\mathcal R_i \colon P(\mathcal B)\to P(\mathcal{A}_i)$ be quantum operations on finite-dimensional von Neumann algebras, with $\mathcal A_i \subseteq \mathcal A\,$, and $\epsilon>0$ such that, for each recovery channel $\mathcal R_i$, we have $\lVert \mathcal R_i \circ \mathcal N [\rho] - \rho|_{\mathcal A_i} \lVert_1 \leq \epsilon\,$ for all $\rho\in S(\mathcal A)$.
  Then:
  \begin{equation}\label{eq:approx mul}
    \bigl\lVert \mathcal N^*[\prod_{i=1}^{n}\mathcal R_i^*[\phi_i] ] - \prod_{i=1}^n \phi_i \bigr\rVert \leq \frac{1}{2} \epsilon \,n (3 n - 1) \prod_{i=1}^n \lVert\phi_i\rVert
  \end{equation}
  for all $\phi_i \in \mathcal A_i$.
\end{thm}
\begin{proof}
  Let $\mathcal A \subseteq B(\mathcal H_A)$ and $\mathcal B \subseteq B(\mathcal H_B)$. We define the operators $V: \mathcal H_A \to \mathcal H_B \otimes \mathcal H_E$ and $W_i: \mathcal H_B \to \mathcal H_A \otimes \mathcal H_{E'}$ to be Stinespring dilations (as in \cref{lem:stinespring}) of $\mathcal{N} = [ V(\mathcal E_{\mathcal A} (\cdot)) V^\dagger ]\big\vert_{\mathcal B}$ and $\mathcal{R}_i = [ W_i(\mathcal E_{\mathcal B} (\cdot)) W_i^\dagger ]\big\vert_{\mathcal A_i}$ respectively. Since $\mathcal N$ and $\mathcal R_i$ are quantum operations,
  \begin{align*}
  \lVert V \rVert, \lVert W_i \rVert  \leq 1.
  \end{align*}
   Let  $T_i = \sqrt{\id -V V^\dagger} W_i^\dagger \phi_i W_i V$. Then
  \[ 0
  \leq  T_i^\dagger T_i
  \leq V^\dagger W_i^\dagger \phi_i^\dagger (\id - W_i V V^\dagger W_i^\dagger) \phi_i  W_i V
  = \mathcal N^*[\mathcal R_i^*[\phi_i^\dagger \phi_i]] - {\mathcal N^*[\mathcal R_i^*[\phi_i]]}^\dagger \mathcal N^*[\mathcal R^*[\phi_i]]
  = \Delta_1 + \Delta_2,
  \]
  where $\Delta_1 := \mathcal N^*[\mathcal R_i^*[\phi_i^\dagger \phi_i]] - \phi_i^\dagger \phi_i$ and $\Delta_2 := \phi_i^\dagger \phi_i - {\mathcal N^*[\mathcal R_i^*[\phi_i]]}^\dagger \mathcal N^*[\mathcal R_i^*[\phi_i]]$.
  By assumption,
  \begin{align*}
    \lVert \mathcal N^*[\mathcal R_i^*[\chi]] - \chi \rVert \leq \epsilon \lVert \chi \rVert,
  \end{align*}
 for all operators $\chi \in \mathcal A_i$. Hence we can bound $\lVert \Delta_1\rVert \leq \epsilon \lVert \phi \rVert^2$ and $\lVert \Delta_2\rVert \leq 2 \epsilon \lVert \phi \rVert^2$
  and it follows that
  \begin{equation}\label{eq:tnormsmall}
    \lVert T_i \rVert  \leq \sqrt{3 \epsilon} \lVert\phi_i\rVert.
  \end{equation}
  As a result, we see that
  \begin{equation*}
  \begin{aligned}
	&\quad\left\lVert V^\dagger  \textstyle \left[\prod_{i=1}^k W_i^\dagger \phi_i W_i \right]\sqrt{(1 - V V^\dagger) }\right\rVert\\
	&\leq \,\textstyle\left\lVert  V^\dagger \left[\prod_{i=1}^{k-1} W_i^\dagger \phi_i W_i\right] V V^\dagger W_k^\dagger\phi_k W_k \sqrt{1 - V V^\dagger} \right\rVert + \left\lVert V^\dagger \left[\prod_{i=1}^{k-1} W_i^\dagger \phi_i W_i\right] (1 -V V^\dagger )W_k^\dagger\phi_k W_k \sqrt{1 - V V^\dagger}\right\rVert\\
	&\leq \, \textstyle\lVert T_k^\dagger \rVert \prod_{i=1}^{k-1} \lVert \phi_i \rVert +  \left\lVert V^\dagger \left[\prod_{i=1}^{k-1} W_i^\dagger \phi_i W_i\right] \sqrt{1 - V V^\dagger}\right\rVert \lVert \phi_k \rVert
  \,\leq\, \textstyle \sqrt{3 \epsilon} \,k\, \prod_{i=1}^{k} \lVert \phi_i \rVert,
  \end{aligned}
  \end{equation*}
  where the first inequality is the triangle inequality, the second follows from the submultiplicativity of the operator norm, and the last inequality uses induction and \eqref{eq:tnormsmall}. Hence
  \begin{equation*}
  \begin{aligned}
   &\quad\textstyle\lVert V^\dagger \left[\prod_{i=1}^k W_i^\dagger \phi_i W_i\right]  (1 - V V^\dagger) W_{(k+1)}^\dagger \phi_{(k+1)} W_{(k+1)} V \prod_{i=k+2}^n V^\dagger W_{i}^\dagger \phi_{i} W_i V\rVert
   \\&\leq \, \lVert V^\dagger  \textstyle \left[\prod_{i=1}^k W_i^\dagger \phi_i W_i \right]\sqrt{(1 - V V^\dagger) }\rVert  \lVert T_{k+1} \rVert \prod_{i=k+2}^n\lVert \phi_i \rVert
   \leq \, 3 k \epsilon \prod_{i=1}^n \lVert \phi_i \rVert,
	\end{aligned}
	\end{equation*}
	and
	\begin{equation}\label{eq:approxmul1}
	\begin{aligned}
	&\quad\textstyle\bigl\lVert \mathcal N^*[\prod_{i=1}^{n}\mathcal R_i^*[\phi_i] ] - \prod_{i=1}^n  \mathcal N^*[\mathcal R_i^*[ \phi_i]] \bigr\rVert
	=\textstyle\, \bigl\lVert V^\dagger \left[\prod_{i=1}^{n} W_i^\dagger \phi_i W_i\right]  V - \prod_{i=1}^n V^\dagger W_i^\dagger  \phi_i W_i V \bigr\rVert\\
	&\leq \textstyle \sum_{k=1}^{n-1} \lVert V^\dagger \left[\prod_{i=1}^k W_i^\dagger \phi_i W_i \right] (1 - V V^\dagger) W^\dagger_{(k+1)} \phi_{k+1} W_{k+1} V \left[\prod_{j = k+2}^{n} V^\dagger W_j^\dagger \phi_j W_j V\right]\rVert
	 \leq  \frac{3}{2} \epsilon n (n-1).
	\end{aligned}
	\end{equation}
	Finally,
	\begin{align}\label{eq:approxmul2}
 	\bigl\lVert  \prod_{i=1}^n  \mathcal N^*[\mathcal R_i^*[ \phi_i]] - \prod_i \phi_i \bigr\rVert \leq \sum_{i=1}^n \lVert \mathcal N^*[\mathcal R_i^*[ \phi_i]] - \phi_i \rVert \prod_{j \neq i} \lVert \phi_j \rVert \leq n \epsilon \prod_{i=1}^n \lVert \phi_i \rVert.
 	\end{align}
 	From \eqref{eq:approxmul1} and \eqref{eq:approxmul2}, we see that \eqref{eq:approx mul} follows immediately by the triangle inequality.
\end{proof}

Although we proved this result for general quantum operations on finite-dimensional von Neumann algebras, the special case of ordinary quantum channels between density matrices of finite-dimensional Hilbert spaces follows immediately by considering $\mathcal A_i = \mathcal A = B(\mathcal H_A)$ and $\mathcal B = B(\mathcal H_B)$. Perhaps surprisingly, the result does not appear to have been previously known, even for a single recovery channel and for ordinary subspace quantum error correction.

To conclude this section, we provide a differential formula for the Choi operator of the recovery map $\mathcal{R}$, which was given in eqn.~\eqref{qBayesDerivative} above:
\begin{equation*}
\Phi_{\mathcal{R}} = \left.\frac{d}{dt}\right|_{t = 0} \log(\overline{\mathcal{N}[\sigma]} \otimes \sigma^{-1} + t \, \Phi_{\mathcal{N}^*}),
\end{equation*}
where $\mathcal{N}^*$ denotes the adjoint of the channel $\mathcal{N}$, and $\overline{\mathcal{N}[\sigma]}$ is the complex conjugate of $\mathcal{N}[\sigma]$.  First recall the integral formula (see, for example,~\cite[Lemma~3.4]{Sutter2016})
\begin{equation}
\label{eq:integral formula}
\left.\frac{d}{dt}\right|_{t = 0} \log(A + t \, B) = \int_{-\infty}^\infty dt\,\beta_0(t) \, A^{-\frac{1}{2} + \frac{it}{2}} B A^{- \frac{1}{2} - \frac{it}{2}}.
\end{equation}
Then we have by direct computation
\begin{align*}
&\left.\frac{d}{dt}\right|_{t = 0} \log(\overline{\mathcal{N}[\sigma]} \otimes \sigma^{-1} + t \, \Phi_{\mathcal{N}^*}) \\
= &\int_{-\infty}^\infty dt\,\beta_0(t) \, {\left(\overline{\mathcal{N}[\sigma]} \otimes \sigma^{-1}\right)}^{-\frac{1}{2} + \frac{it}{2}} \Phi_{\mathcal{N}^*} {\left(\overline{\mathcal{N}[\sigma]} \otimes \sigma^{-1}\right)}^{- \frac{1}{2} - \frac{it}{2}} \\
= &\int_{-\infty}^\infty dt\,\beta_0(t) \, (I \otimes \sigma^{\frac{1}{2} - \frac{it}{2}}) \left({\left(\overline{\mathcal{N}[\sigma]} \otimes I \right)}^{-\frac{1}{2} + \frac{it}{2}}\,(I \otimes \mathcal{N}^*)[\Phi] \, {\left(\overline{\mathcal{N}[\sigma]}\otimes I \right)}^{- \frac{1}{2} - \frac{it}{2}} \right) \,(I \otimes \sigma^{\frac{1}{2} + \frac{it}{2}}) \\
= &\int_{-\infty}^\infty dt\,\beta_0(t) \,(I \otimes \sigma^{\frac{1}{2} - \frac{it}{2}}) \, (I \otimes \mathcal{N}^*)\left[(I \otimes {\mathcal{N}[\sigma]}^{-\frac{1}{2} + \frac{it}{2}})\, \Phi \, (I \otimes {\mathcal{N}[\sigma]}^{- \frac{1}{2} - \frac{it}{2}})\right] \,(I \otimes \sigma^{\frac{1}{2} + \frac{it}{2}}) \\
\end{align*}
which, comparing to eq.~\eqref{eq:universal recovery map vN}, is indeed $\Phi_{\mathcal{R}}$.

\section{Entanglement Wedge Reconstruction for Algebras}
We are interested in reconstructing bulk operators acting on the entanglement wedge of a subregion of the boundary CFT using only boundary data supported in that subregion.
A simplified picture of our setup will include the following data:
a boundary CFT modeled by an algebra of observables $\mathcal M_{\text{CFT}}$,
a subalgebra $\mathcal M_A \subseteq \mathcal M_{\text{CFT}}$ of operators acting on a boundary subregion $A$ of the CFT,
a code space modeled by an algebra of bulk observables $\mathcal M_{\text{code}}$, and
a subalgebra $\mathcal M_a \subseteq \mathcal M_{\text{code}}$ of operators acting inside the entanglement wedge of $A$.
We also have a bulk-to-boundary map $\mathcal J\colon S(\mathcal M_{\text{code}}) \to S(\mathcal M_{\text{CFT}})$ taking code states in the bulk to states on the boundary.
The setup is as follows:
\begin{center}
	\begin{tikzcd}[column sep=5em, row sep=3em]
	  \mathcal M_a \arrow[r, yshift=0.7ex, hookrightarrow, "\text{incl}"]
	  \arrow[d, dashed, "\mathcal R^*", swap] &
	  \mathcal M_\text{code} \\
	  \mathcal M_A \arrow[r, swap, hookrightarrow, "\text{incl}"] &
	  \mathcal M_\text{CFT} \arrow[u, swap, "\mathcal J^*"]
	\end{tikzcd}
	\qquad\qquad
	\begin{tikzcd}[column sep=3em, row sep=3em]
	  S(\mathcal M_a) 
	  &
	  \arrow[l, yshift=0.7ex, "\text{res}", swap] S(\mathcal M_\text{code}) \arrow[d, "\mathcal J"] \\
	  S(\mathcal M_A) \arrow[u, dashed, "\mathcal R"] &
	  \arrow[l, "\text{res}"] S(\mathcal M_\text{CFT})
	\end{tikzcd}
\end{center}
The map $\mathcal R^*$ (dashed) is the desired map implementing entanglement wedge reconstruction that we will construct in \cref{thm:algebraic entanglement wedge} below.
A fully general treatment of the problem would include infinite-dimensional algebras of observables.
However, there are many technical difficulties in infinite dimensions, and as such we will restrict ourselves to finite-dimensional algebras as in~\cite{Harlow2016}.
Our setup and analysis will closely resemble the one used in the main body of this document, with appropriate changes made to account for the more general algebraic structure.

The following lemma generalizes our main results to this setup, showing that approximate equality of relative entropies implies approximate entanglement wedge reconstruction even at the level of algebras.

\begin{thm}\label{thm:algebraic entanglement wedge}
  Let $\mathcal M_a\subseteq\mathcal M_\text{code}$ and $\mathcal M_A\subseteq\mathcal M_\text{CFT}$ be finite-dimensional von Neumann algebras, $\mathcal J\colon S(\mathcal M_\text{code}) \to S(\mathcal M_\text{CFT})$ a quantum channel, and $\epsilon>0$ such that
	\begin{equation}\label{JLMSeqn}
	  \abs{ D(\rho_a\Vert\sigma_a) - D({\mathcal J[\rho]}_A\Vert{\mathcal J[\sigma]}_A) } \leq \epsilon
	\end{equation}
	for all $\rho,\sigma\in S(\mathcal M_\text{code})$, where we denote by $\rho_X$ the restriction $\rho\vert_{\mathcal M_X}$ of a state $\rho$ to some subalgebra $\mathcal M_X$.
	Then there exists a map $\mathcal R\colon S(\mathcal M_A)\to S(\mathcal M_a)$ such that, for all $\rho\in S(\mathcal M_\text{code})$ and $\phi_a,\phi'_a\in\mathcal M_a$,
	\begin{enumerate}[label={(\roman*)}]
	\item\label{item:i} ${\lVert \rho_a - \mathcal R[{\mathcal J[\rho]}_A] \rVert}_1 \leq \delta$,
	\item\label{item:ii} $\bigl\lvert \Expect{\mathcal R^*[\phi_a]}_{\mathcal J[\rho]} - \Expect{\phi_a}_\rho \bigr\rvert \leq \delta \lVert\phi_a\rVert$,
	\item\label{item:iii} $\bigl\lvert \Expect{\prod_{i=1}^n \mathcal R^*[\phi_i]}_{\mathcal J[\rho]} - \Expect{\prod_{i=1}^n \phi_i}_\rho \bigr\rvert \leq (\sqrt{2 \ln 2} + (3 n - 1) n)\, \sqrt{\epsilon}\prod_i \lVert \phi_i \rVert$,
	\end{enumerate}
	where $\delta := (2 + \sqrt{2\ln 2}) \sqrt\epsilon$.
	Explicitly,
	\begin{equation}\label{eq:explicit bulk recovery}
	  \mathcal R^*[\phi_a] = \int dt \, \beta_0(t) \, e^{\frac{1-it}2 H_A} \mathcal J{\bigl[\mathcal E_a[e^{-\frac{1-it}2 H_a} \phi_a e^{-\frac{1+it}2H_a}]\bigr]}_A \, e^{\frac{1+it}2 H_A},
	\end{equation}
	where $H_a = -\log\sigma_a$ and $H_A = -\log {\mathcal J[\mathcal E_a[\sigma_a]]}_A$ for some arbitrary fixed full-rank state $\sigma_a\in S(\mathcal M_a)$, with $\mathcal E_a$ the state extension map $S(\mathcal M_a)\subseteq S(\mathcal M_\text{code})$ from~\eqref{eq:state extension}.
\end{thm}
\begin{proof}
  We consider the ``local'' quantum channel
  \[ \mathcal N\colon S(\mathcal M_a)\to S(\mathcal M_A), \quad \omega_a \mapsto {\mathcal J[\mathcal E_a[\omega_a]]}_A. \]
  A crucial property of $\mathcal N$ is that, for any state $\rho\in S(\mathcal M_\text{code})$,
  \begin{equation}\label{eq:crucial property}
  \begin{aligned}
  	&\lVert \mathcal N[\rho_a] - {\mathcal J[\rho]}_A \rVert_1
  = \lVert {\mathcal J[\mathcal E_a[\rho_a]]}_A - {\mathcal J[\rho]}_A \rVert_1
  \leq \sqrt{2\ln 2\, D({\mathcal J[\mathcal E_a[\rho_a]]}_A \Vert {\mathcal J[\rho]}_A)} \\
  &\qquad \leq \sqrt{2\ln 2\, (D({\mathcal E_a[\rho_a]}_a \Vert \rho_a) + \epsilon)}
  = \sqrt{2\ln 2\, \epsilon},
  \end{aligned}
  \end{equation}
  where the first inequality is the familiar relation between trace norm and relative entropy; the second inequality is our assumption~\eqref{JLMSeqn}, and the last identity is~\eqref{eq:state extension}.

  We now let $\mathcal R=\mathcal R_{\sigma_a,\mathcal N}$ denote the recovery map~\eqref{eq:universal recovery map vN} associated with some full-rank state $\sigma_a\in S(\mathcal M_a)$.
  With this choice of $\mathcal R$, \cref{eq:explicit bulk recovery} holds true.
  Moreover, \cref{lem:universal recovery} shows that, for all $\rho\in S(\mathcal M_\text{code})$,
  \[ -2\log F(\rho_a, (\mathcal R \circ \mathcal N)[\rho_a])
  \leq D(\rho_a\Vert\sigma_a) - D(\mathcal N[\rho_a]\Vert\mathcal N[\sigma_a])
  = D({\mathcal E_a[\rho_a]}_a\Vert{\mathcal E_a[\sigma_a]}_a) - D({\mathcal J[\mathcal E_a[\rho_a]]}_A\Vert{\mathcal J[\mathcal E_a[\sigma_a]]}_A)
  \leq \epsilon \]
  by using our assumption~\eqref{JLMSeqn} and~\eqref{eq:state extension} a second time.
  Thus, $F(\rho_a, (\mathcal R \circ \mathcal N)[\rho_a]) \geq 1 - \epsilon/2$, and using the Fuchs-van de Graaf inequality,
  \begin{equation}\label{eq:fuchs recovery}
  	\lVert \rho_a - \mathcal R[\mathcal N[\rho_a]] \rVert_1
  	\leq 2\sqrt{1 - F{(\rho_a, (\mathcal R \circ \mathcal N)[\rho_a])}^2}
	  \leq 2\sqrt{\epsilon}.
  \end{equation}
  We obtain~\ref{item:i} from~\cref{eq:crucial property,eq:fuchs recovery} and the triangle inequality.
  This readily implies~\ref{item:ii}, since
  \begin{align*}
    \abs{\Expect{\mathcal R^*[\phi_a]}_{\mathcal J[\rho]} - \Expect{\phi_a}_\rho}
  = \abs{\Expect{\mathcal R^*[\phi_a]}_{{\mathcal J[\rho]}_A} - \Expect{\phi_a}_{\rho_a}}
  = \abs{\Expect{\phi_a}_{\mathcal R[{\mathcal J[\rho]}_A]} - \Expect{\phi_a}_{\rho_a}}
    \leq \delta \lVert\phi_a\rVert.
  \end{align*}
  For~\ref{item:iii}, we observe that
  \begin{align*}
	&\textstyle\quad \abs{\Expect{\prod_i \mathcal R^*[\phi_i] }_{\mathcal J[\rho]} - \Expect{\prod_i \phi_i}_\rho}
	= \abs{\Expect{\prod_i \mathcal R^*[\phi_i] }_{{\mathcal J[\rho]}_A} - \Expect{\prod_i \phi_i}_{\rho_a}} \\
	&\textstyle\leq \abs{\Expect{\prod_i \mathcal R^*[\phi_i] }_{\mathcal N[\rho_a]} - \Expect{\prod_i \phi_i}_{\rho_a}} + \sqrt{2\ln 2\, \epsilon}\,\, \prod_i \lVert \phi_i \rVert \\
	&\textstyle= \abs{\Expect{\mathcal N^*[\prod_i \mathcal R^*[\phi_i] ]}_{\rho_a} - \Expect{\prod_i \phi_i}_{\rho_a}} + \sqrt{2\ln 2\, \epsilon}\,\, \prod_i \lVert \phi_i \rVert \\
	&\textstyle\leq  (\sqrt{2\ln 2} + (3  n - 1) n )\,\sqrt{\epsilon}\, \prod_i \lVert \phi_i \rVert
  \end{align*}
  where the first inequality is~\eqref{eq:crucial property} and the second inequality is \cref{thm:approx mul}.
\end{proof}
When the fiducial state $\sigma_a\in S(\mathcal M_a)$ is chosen to be the maximally mixed state then the map~\eqref{eq:explicit bulk recovery} takes a particularly simple form.
In this case, $\mathcal E_a[\sigma_a]=\id_\text{code}/d_\text{code}=:\tau_\text{code}$, where $d_\text{code}:=\tr[\id_\text{code}]$.
We obtain
\begin{equation}\label{eq:simplified algebraic}
  \mathcal R^*[\phi_a] = \frac1{d_\text{code}} \int dt \, \beta_0(t) \, e^{\frac{1-it}2 H_A} {\mathcal J[\phi_a]}_A \, e^{\frac{1+it}2 H_A},
\end{equation}
where we recall that $H_A = -\log {\mathcal J[\tau_\text{code}]}_A$ is the modular Hamiltonian of the boundary region $A$ associated with the maximally mixed code state.
\Cref{eq:simplified algebraic} can be rewritten as follows (\cite[Lemma 3.4]{Sutter2016}),
\begin{equation}\label{eq:simplified algebraic derivative}
  \mathcal R^*[\phi_a]
= \frac1{d_\text{code}} \frac{d}{dt}\bigg\vert_{t=0} \log {\mathcal J[\tau_\text{code} + t \phi_a]}_A
= -\frac1{d_\text{code}} \frac{d}{dt}\bigg\vert_{t=0} H_A[\tau_\text{code} + t \phi_a]
\end{equation}
where we have introduced the notation $H_A[\rho] := -\log {\mathcal J[\rho]}_A$ for the boundary modular Hamiltonian on subregion $A$ associated with a bulk state $\rho\in S(\mathcal M_\text{code})$.
That is, the boundary operator $\mathcal R^*[\phi_a]$ can be found as the response of the boundary modular Hamiltonian to a perturbation of the fiducial bulk state in the direction of the bulk operator $\phi_a$.

If $\sigma_a$ is \textit{not} the maximally mixed state but is instead some arbitrary state, \cref{eq:simplified algebraic derivative} for $\mathcal{R}^*[\phi_a]$ no longer holds.  However, using \cref{eq:integral formula} we can write down a similar equation for the Choi operator of $\mathcal{R}^*$ itself:
\begin{equation*}
\Phi_{\mathcal{R}^*} = \left.\frac{d}{dt}\right|_{t = 0} \log(\overline{\sigma}_a^{-1} \otimes {\mathcal{J}[\mathcal{E}_a[\sigma_a]]}_A + t \, \Phi_{{\mathcal{J}[\mathcal{E}_a[\cdot]]}_A}),
\end{equation*}
where $\Phi_{{\mathcal{J}[\mathcal{E}_a[\cdot]]}_A}$ is the Choi operator of ${\mathcal{J}[\mathcal{E}_a[\cdot]]}_A$.

%

Finally we show that correlation functions of bulk operators are preserved even when each operator is reconstructed using a different entanglement wedge.
\begin{corr}\label{corr:multientwedge}
Let $\mathcal M_{a_i}\subseteq\mathcal M_\text{code}$ and $\mathcal M_{A_i}\subseteq\mathcal M_\text{CFT}$ be sets of finite-dimensional von Neumann algebras, $\mathcal J\colon S(\mathcal M_\text{code}) \to S(\mathcal M_\text{CFT})$ a quantum channel such that, for each pair of algebras $\mathcal M_{a_i}$ and $\mathcal M_{A_i}$, the JLMS condition \eqref{JLMSeqn} holds for some $\epsilon > 0$. Then
\begin{align}\label{eq:multientwedge}
\textstyle\bigl\lvert\, \Expect{\prod_{i=1}^n \mathcal R_{A_i}^*[\phi_i]}_{\mathcal J[\rho]} - \Expect{\prod_{i=1}^n \phi_i}_\rho \bigr\rvert \leq \frac{1}{2} n (3n - 1)  (2 + \sqrt{2\ln 2}) \sqrt\epsilon \prod_{i=1}^n \lVert\phi_i\rVert,
\end{align}
where the recovery maps $\mathcal R_{A_i}$ are defined by applying the explicit construction \eqref{eq:explicit bulk recovery} to the pairs of algebras $\mathcal M_{a_i}$ and $\mathcal M_{A_i}$.
\end{corr}
\begin{proof}
The proof is a simple consequence of applying \cref{thm:approx mul} to condition \ref{item:i} of \cref{thm:algebraic entanglement wedge}. If we take $\mathcal J$ to be the encoding map (denoted as $\mathcal N$ in \cref{thm:approx mul}) and include the restriction onto each boundary subalgebra $\mathcal M_{A_i}$ as part of the corresponding recovery map $\mathcal R_{A_i}$, then \eqref{eq:multientwedge} follows immediately.
\end{proof}

\section{Rindler Wedge Reconstruction from Global Reconstruction for Free Fields}
In this section, we work through an illustrative example by applying our entanglement wedge reconstruction formula in a problem motivated by AdS-Rindler reconstruction using only global HKLL~\cite{HKLL} as input.  This example is only strictly valid for free field theories, but we will nevertheless use the language of AdS/CFT for familiarity; in short, we pretend both bulk and boundary fields are free, and we comment on the difficulties that arise when the boundary field is only a generalized free field.
Let $A$ be the boundary region corresponding to a single Rindler wedge, $a$ be the entanglement wedge of $A$, $D_A$ be the boundary domain of dependence of $A$, and $\bar A$ be the complement of $A$.  We will show that local bulk operators in $a$ can be represented as linear combinations of field operators for Rindler modes confined to region $A$.  The subtlety for a true AdS/CFT calculation lies in the Rindler decomposition -- in general, no such decomposition exists for generalized free fields.

For simplicity, we consider the case of AdS$_3$.  We use Poincar\'{e} patch coordinates,
\begin{equation*}
ds^2 = \frac{\ell^2}{z^2} (-dt^2 + dx^2 + dz^2),
\end{equation*}
and we label bulk points by $Y=(t,x,z)$ and boundary points by $y=(t,x)$.

Suppose that we want to reconstruct a bulk operator $\phi(Y)$ for $Y \in a$ on the boundary of the Rindler wedge $A$.  Let us denote the vacuum state by $\ket{\widetilde{0}}$, and an excitation of the ground state by $\ket{\widetilde{1}} = \phi(Y)\ket{\widetilde{0}}$, which we take to be normalized.
We will consider a two-dimensional code space $\CH_\text{code}=\text{span}\{\ket{\widetilde{0}},\ket{\widetilde{1}}\}$.  Our goal will be to reconstruct the action of the operator $\phi(Y)$ on the code space, restricted to the boundary interval $A$.
The maximally mixed state on the code is simply $\tau = \frac{1}{2}(\proj{\widetilde{0}}+\proj{\widetilde{1}})$. Note that in this simple example $\HS_a=\HS_\text{code}$, since we do not consider any degrees of freedom living in $\bar a$.

As an operator on the code space, $\phi(Y)$ acts as $X:=\ketbra{\widetilde{1}}{\widetilde{0}} + \mathit{h.c.}$,
mapping the vacuum state to the excited state and vice versa.  With our chosen operator, code space, and maximally mixed state, we can rewrite \cref{ExplicitRecoveryHilbert} as
\begin{equation}
\label{adjointB2B2}
\CR^*[X] = \frac{1}{2}\int dt \,\beta_0(t)\, {\CN[\tau]}^{\frac{-1+it}{2}} \CN\Bigl[\ketbra{\widetilde{1}}{\widetilde{0}} + \ketbra{\widetilde{0}}{\widetilde{1}}\Bigr] {\CN[\tau]}^{-\frac{1+it}{2}},
\end{equation}
where we have introduced $\CN[\rho_a]=\tr_{\bar{A}} \bigl[ J \rho_a J^\dagger \bigr]$ as a shorthand.

In order to evaluate \cref{adjointB2B2}, we will need to compute terms of the form
\begin{equation*}
\mathcal{N}\bigl[\ketbra{\widetilde{x}}{\widetilde{y}}\bigr] = \tr_{\bar A}\bigl[\ketbra{x}{y}\bigr],
\end{equation*}
where $x,y\in \{0,1\}$, and the states $\ket{x}:=J\ket{\widetilde{x}}$.  The empty AdS state $\ket{\widetilde{0}}$ is mapped to the CFT ground state $J\ket{\widetilde{0}} = \ket{0}$.  The excited state $\ket{\widetilde{1}}$ is mapped via global HKLL to
\begin{equation*}
J \ket{\widetilde{1}} =:  \ket{1} = \int_{y' \in D} dy'\, K_g(Y,y')\, \Phi(y') \ket{0}\,
\end{equation*}
where $K_g$ is a bulk-to-boundary kernel ($g$ denotes ``global''), $\Phi(y)$ is a boundary operator, and $D$ is a boundary spacetime domain.

To leading order in $1/N$, $\Phi(y)$ behaves like a generalized free field \cite{Greenberg1961}.  Unfortunately, generalized free fields do not, in general, admit a decomposition into Rindler modes.  Thus, we now pretend the boundary field is instead a true free field, and we expand it in terms of Rindler modes $\{a_\ell, b_\ell\}$ adapted to $A$ and $B = \bar A$\,:
\begin{equation}
\label{RindlerExpand1}
\Phi(y) = \sum_\ell f_{a,\ell}(y) a_\ell + f_{a,\ell}^*(y) a^\dagger_\ell + f_{b,\ell}(y) b_\ell + f_{b,\ell}^*(y) b^\dagger_\ell.
\end{equation}
In order to upgrade this calculation to a true AdS/CFT computation, some care would need to be taken with respect to this decomposition for generalized free fields, but we nevertheless forge ahead in the name of pedagogy.  Note that the modes $a_\ell$ and $b_\ell$ are entangled in the state $\ket{0}$, since there is entanglement between $A$ and $B$.

With the HKLL prescription of $\ket{0} = J\ket{\widetilde{0}}$ and $\ket{1} = J\ket{\widetilde{1}}$ in hand, we can now compute the partial trace of the various matrix elements appearing in \cref{adjointB2B2} with respect to the region $\bar A$.  We begin with
$\tr_B\bigl[\proj{0} + \proj{1}\bigr]$.  Since $\ket{0}$ and $\ket{1}$ are approximately distinguishable on $A$, the result of tracing out $\bar A$ (\emph{i.e.}, tracing out the $B$ modes) will be approximately block diagonal. On the upper block of the reduced density matrix we have the ground state density matrix $\rho_{A,0} = \text{Tr}_B\bigl[\proj{0}\bigr]$ for $A$, and on the lower block we have
\begin{equation*}
\label{LowerBlock1}
\rho_{A,1}= \text{Tr}_{\bar A}\bigl[\proj{1}\bigr] = \text{Tr}_{\bar A}\left[\int_{y_1 \in D} \int_{y_2 \in D} K_g(Y,y_1) K_g(Y,y_2) \Phi(y_1) \ket{0} \bra{0} \Phi(y_2) \right].
\end{equation*}

To simplify the form of $\rho_{A,1}$, a crucial fact is that the ground state $\ket{0}$ is comprised of entangled $a_\ell$ and $b_\ell$ modes.
More precisely, focusing on a single mode, there is a modular energy $E_\ell$ such that the ground state is of the form
\begin{equation*}
\ket{0} \propto \sum_n e^{- \pi E_\ell n} \ket[a_\ell]{n} \ket[b_\ell]{n}.
\end{equation*}
In such a state the reduced density matrix of $A$ has the form $\rho_{A,0} \propto \sum_n e^{- 2 \pi E_\ell n} \proj[a_\ell]{n}$, so that
\begin{equation*}
\ket{0} \propto \rho_{A,0}^{1/2} \sum_{n} |n\rangle_{a_\ell} |n\rangle_{b_\ell}.
\end{equation*}
Using the ``transpose trick'', it follows that mode operators on $\bar A$ can be written in terms of operators on $A$:
\begin{equation}\label{transposetrickformodes}
\begin{aligned}
b_\ell \ket{0} &= \rho_{A,0}^{1/2} a_\ell^\dagger \rho_{A,0}^{-1/2} \ket{0},\\
b_\ell^\dagger \ket{0} &= \rho_{A,0}^{1/2} a_\ell \rho_{A,0}^{-1/2} \ket{0},
\end{aligned}
\end{equation}
and vice versa.  We also note the following helpful identities:
\begin{equation}\label{helpfulidentities}
\begin{aligned}
\rho_{A,0}^{1/2} a_\ell \rho_{A,0}^{-1/2} &= a_\ell e^{\pi E_\ell},\\
\rho_{A,0}^{-1/2} a_\ell \rho_{A,0}^{1/2} &= a_\ell e^{-\pi E_\ell},\\
\rho_{A,0}^{1/2} a_\ell^\dagger \rho_{A,0}^{-1/2} &= a_\ell^\dagger e^{-\pi E_\ell},\\
\rho_{A,0}^{-1/2} a_\ell^\dagger \rho_{A,0}^{1/2} &= a_\ell^\dagger e^{\pi E_\ell}.
\end{aligned}
\end{equation}
Using \cref{RindlerExpand1}, we now evaluate $\CN\bigl[\ketbra{\widetilde{1}}{\widetilde{0}}\bigr] = \text{Tr}_{\bar A}\bigl[\, \ketbra{1}{0}\,\bigr]$:
\begin{equation*}
\CN[\,\ketbra{\widetilde{1}}{\widetilde{0}}\,] = \int_{y\in D} K_g(Y,y) \,\text{Tr}_{\bar A}\left[\left(\sum_\ell f_{a,\ell}(y) a_\ell + f_{a,\ell}^*(y) a^\dagger_\ell + f_{b,\ell}(y) b_\ell + f_{b,\ell}^*(y) b^\dagger_\ell \right) \proj{0}\right]
\end{equation*}
For brevity, let
\begin{align*}
\check{f}_{a,\ell} &= \int_{y \in D} K_g(Y,y) f_{\ell,a}(y), \\
\check{f}_{b,\ell} &= \int_{y \in D} K_g(Y,y) f_{\ell,b}(y).
\end{align*}
Using \cref{transposetrickformodes}, we can rewrite $\CN\bigl[\ketbra{\widetilde{1}}{\widetilde{0}}\bigr]=\int K_g \Phi \proj{0} = Q_A \proj{0}$, where
\begin{equation*}
Q_A = \sum_\ell \check{f}_{a,\ell} a_\ell + \check{f}_{a,\ell}^* a_\ell^\dagger + \check{f}_{b,\ell} \rho_{A,0}^{1/2} a_\ell^\dagger \rho_{A,0}^{-1/2} + \check{f}_{b,\ell}^* \rho_{A,0}^{1/2} a_\ell \rho_{A,0}^{-1/2}.
\end{equation*}
Finally, using the identities~\eqref{helpfulidentities}, we write $Q_A$ as
\begin{equation*}
Q_A = \sum_\ell \check{f}_{a,\ell} a_\ell + \check{f}_{a,\ell}^* a_\ell^\dagger + \check{f}_{b,\ell} a_\ell^\dagger e^{-\pi E_\ell} + \check{f}_{b,\ell}^*  a_\ell e^{\pi E_\ell},
\end{equation*}
taking note that this operator is only guaranteed to reproduce the action of $\int K_g \Phi$ when acting on $\ket{0}$.

Using $Q_A$ we can write $\CN[\ketbra{\widetilde{x}}{\widetilde{y}}] = {(Q_A)}^x \rho_{A,0} {(Q_A^\dagger)}^y$, where $x,y \in \{0,1\}$. In particular, we have that $\rho_{A,1} = Q_A \rho_{A,0} Q_A^\dagger$ and $\CN[\,\ket{\widetilde{1}} \bra{\widetilde{0}}\,] = Q_A \rho_{A,0}$.  We then rewrite $Q_A \rho_{A,0}=\rho_{A,0}^{1/2} \left(\rho_{A,0}^{-1/2} Q_A \rho_{A,0}^{1/2} \right) \rho_{A,0}^{1/2}$,
where the term in the parenthesis is still a simple sum of $a_\ell$ and $a_\ell^\dagger$ with various weights. Explicitly, it is
\begin{equation*}
\rho_{A,0}^{-1/2} Q_A \rho_{A,0}^{1/2} = \sum_\ell \check{f}_{a,\ell} a_\ell e^{-\pi E_\ell} + \check{f}_{a,\ell}^* a_\ell^\dagger e^{\pi E_\ell} + \check{f}_{b,\ell} a_\ell^\dagger + \check{f}_{b,\ell}^*  a_\ell.
\end{equation*}
Using the approximate block diagonality of $\CN[\tau]$ we can write
\begin{equation*}
{\CN[\tau]}^{\frac{-1+it}{2}} \approx {\left(\frac{1}{2}\right)}^{\frac{1-it}{2}} \left(\rho_{A,0}^{\frac{-1+it}{2}} + \rho_{A,1}^{\frac{-1+it}{2}}\right),
\end{equation*}
although the reader is cautioned that we have not analyzed the quality of this approximation.
The approximate block diagonality also implies that
\begin{equation*}
{\CN[\tau]}^{\frac{-1+it}{2}} \rho_{A,0}^{1/2} \approx {\left( \frac{1}{2}\right)}^{\frac{1-it}{2}} \rho_{A,0}^{\frac{-1+it}{2}} \rho_{A,0}^{1/2} = {\left( \frac{1}{2} \right)}^{\frac{1-it}{2}} \rho_{A,0}^{it/2},
\end{equation*}
and a similar argument shows that
\begin{equation*}
\rho_{A,0}^{1/2} {\CN[\tau]}^{\frac{-1-it}{2}} \approx {\left( \frac{1}{2} \right)}^{\frac{1+it}{2}} \rho_{A,0}^{-it/2}.
\end{equation*}
Thus, the recovery channel is proportional to
\begin{equation*}
\CR^*[\,\ketbra{\widetilde{1}}{\widetilde{0}}\,] = \int dt \,\beta_0(t)\, \rho_{A,0}^{it/2} \left(\rho_{A,0}^{-1/2} Q_A \rho_{A,0}^{1/2} \right) \rho_{A,0}^{it/2},
\end{equation*}
and we note that the factors of $1/2$ have canceled out. The combined object $\rho_{A,0}^{it/2} \left(\rho_{A,0}^{-1/2} Q_A \rho_{A,0}^{1/2} \right) \rho_{A,0}^{-it/2}$ is then
\begin{equation*}
\sum_\ell \check{f}_{a,\ell} a_\ell e^{-\pi E_\ell+i \pi E_\ell t} + \check{f}_{a,\ell}^* a_\ell^\dagger e^{\pi E_\ell-i \pi E_\ell t} + \check{f}_{b,\ell} a_\ell^\dagger e^{-i \pi E_\ell t} + \check{f}_{b,\ell}^*  a_\ell e^{i \pi E_\ell t}.
\end{equation*}
Defining $\widehat{\beta}_0(\omega) := \int dt \, \beta_0(t) e^{i \omega t}$ and noting that $\widehat{\beta}_0(-\omega)= \widehat{\beta}_0(\omega)$ by symmetry of $\beta_0(t)$, the recovery map acting on our operator is
\begin{equation}
\label{recoveryResult1}
\CR^*[\,\ketbra{\widetilde{1}}{\widetilde{0}}\,] = \sum_\ell \check{f}_{a,\ell} a_\ell e^{-\pi E_\ell} \widehat{\beta}_0(\pi E_\ell) + \check{f}_{a,\ell}^* a_\ell^\dagger e^{\pi E_\ell}\widehat{\beta}_0(\pi E_\ell) + \check{f}_{b,\ell} a_\ell^\dagger \widehat{\beta}_0(\pi E_\ell) + \check{f}_{b,\ell}^*  a_\ell \widehat{\beta}_0(\pi E_\ell),
\end{equation}
and there is an analogous expression for $\CR^*[\,\ketbra{\widetilde{0}}{\widetilde{1}}\,]$.

\Cref{recoveryResult1} is our desired result.The bulk operator $X=\ketbra{\widetilde{1}}{\widetilde{0}}+\textit{h.c.}$ can be reconstructed on the Rindler wedge using only Rindler mode operators. Moreover, only single mode operators appear.
\end{document}